\documentclass[12pt]{article}
\usepackage{enumerate}
\usepackage{natbib}
\usepackage{longtable}
\usepackage{url} 

\newcommand{\blind}{1}

\usepackage{setspace}
\usepackage{etoolbox}
\BeforeBeginEnvironment{equation}{\begin{singlespace}\vspace*{-\baselineskip}\vspace*{-0.2cm}}
\AfterEndEnvironment{equation}{\end{singlespace}\vspace*{-0.2cm}\noindent\ignorespaces}

\BeforeBeginEnvironment{equation*}{\begin{singlespace}\vspace*{-\baselineskip} \vspace*{-0.2cm}}
\AfterEndEnvironment{equation*}{\end{singlespace}\vspace*{-0.2cm}\noindent\ignorespaces}

\BeforeBeginEnvironment{align}{\begin{singlespace}\vspace*{-\baselineskip} \vspace*{-0.2cm}}
\AfterEndEnvironment{align}{\end{singlespace}\vspace*{-0.2cm}\noindent\ignorespaces}

\BeforeBeginEnvironment{align*}{\begin{singlespace}\vspace*{-\baselineskip}\vspace*{-0.2cm}}
\AfterEndEnvironment{align*}{\end{singlespace}\vspace*{-0.2cm}\noindent\ignorespaces}

\AtBeginEnvironment{tabular}{\singlespacing}

\usepackage{titlesec}

\titlespacing\section{0pt}{0pt}{0pt plus 2pt minus 2pt}
\titlespacing\subsection{0pt}{0pt}{0pt plus 2pt minus 2pt}
\titlespacing\subsubsection{0pt}{0pt}{0pt plus 2pt minus 2pt}

\RequirePackage{amsthm,amsmath,amsfonts,amssymb,amsthm}
\RequirePackage[colorlinks,citecolor=blue,urlcolor=blue]{hyperref}
\RequirePackage{graphicx}
\usepackage{ifoddpage}
\usepackage{float}
\usepackage{changepage}

\PassOptionsToPackage{unicode}{hyperref}
\PassOptionsToPackage{naturalnames}{hyperref}

\usepackage{mathrsfs, latexsym}
\usepackage{bbm}
\usepackage[ruled,vlined]{algorithm2e}
\SetKwComment{commentSt}{ }{}
\usepackage{makecell}
\usepackage{mathrsfs}
\usepackage{courier}
\def\nn{\mathbb{N}}

\usepackage[dvipsnames]{xcolor}

\usepackage{appendix}
\renewcommand\appendixname{Supplementary Material}

\def\rr{\mathbb{R}}

\def\nn{\mathbb{N}}

\def\argmin{\text{argmin}}
\def\argmax{\text{argmax}}
\def\1{\mathbbm{1}}

\allowdisplaybreaks
\usepackage{capt-of}

\renewcommand{\leq}{\leqslant} 
\renewcommand{\geq}{\geqslant}

\newcommand{\pluseq}{\mathrel{+}=}

\newcommand{\tr}{\mathbf{Tr}}

\def\qed{ \hfill $\blacksquare$}  

\newtheorem{theorem}{Theorem}[section]
\newtheorem{proposition}[theorem]{Proposition}

\newtheorem{corollary}[theorem]{Corollary}
\theoremstyle{remark}

\begin{document}

\def\spacingset#1{\renewcommand{\baselinestretch}%
{#1}\small\normalsize} \spacingset{1}


\if1\blind
{
  \title{\bf Bayes Watch: Bayesian Change-point Detection for Process Monitoring with Fault Detection}
  \author{Alexander C. Murph\thanks{
    The authors gratefully acknowledge the National Science Foundation under Grant No. DMS-1916115, 2113404, and 2210337, and the National Heart, Lung, and Blood Institute of the National Institutes of Health under Award Number R56HL155373.  } \\
    Statistical Sciences Group, \\
    Computer, Computational, and Statistical Sciences Division,\\ 
    Los Alamos National Laboratory \\
    and \\
    Curtis B. Storlie, Patrick M. Wilson \\
    The Kern Center for the Science of Health Care Delivery,\\ The Mayo Clinic, Rochester, MN \\
    and \\
    Jonathan P. Williams \\
    Department of Statistics, NC State University\\
    Jan Hannig\\
    Department of Statistics \& Operations Research,\\ University of NC at Chapel Hill
    }
  \maketitle
} \fi
\if0\blind
{
  \bigskip
  \bigskip
  \bigskip
  \begin{center}
    {\LARGE\bf Bayes Watch: Bayesian Change-point Detection for Process Monitoring with Fault Detection}
\end{center}
  \medskip
} \fi
\vspace*{-0.25cm}
\begin{abstract}
When a predictive model is in production, it must be monitored in real-time to ensure that its performance does not suffer due to drift or abrupt changes to data.  Ideally, this is done long before learning that the performance of the model itself has dropped by monitoring outcome data.  
In this paper we consider the problem of monitoring a predictive model that identifies the need for palliative care currently in production at the Mayo Clinic in Rochester, MN.  We introduce a framework, called \textit{Bayes Watch}, for detecting change-points in high-dimensional longitudinal data with mixed variable types and missing values and for determining in which variables the change-point occurred.  Bayes Watch fits an array of Gaussian Graphical Mixture Models to groupings of homogeneous data in time, called regimes, which are modeled as the observed states of a Markov process with unknown transition probabilities.  In doing so, Bayes Watch defines a posterior distribution on a vector of regime assignments, which gives meaningful expressions on the probability of every possible change-point.  Bayes Watch also allows for an effective and efficient fault detection system that assesses what features in the data where the most responsible for a given change-point.  
\end{abstract}
\noindent%
{\it Keywords:}  Reversible Jump MCMC, $G$-Wishart, change-point
\vfill

\newpage
\spacingset{1.75}

\section{Introduction}\label{sec:intro}

The tremendous amount of medical data available from hospitals combined with modern advancements in machine learning methods have a lot of promise in 
improving today's quality of care and extending patient survival.  One success of a data-driven approach to improving quality of care for patients is the work of \cite{murphree2021}, where the Mayo Clinic developed an in-process Gradient Boosting Machine (GBM) to predict the chance of individual patients needing a palliative care consultation within the next 7 day period.  Today, this model still sees active use at the Mayo Clinic in Rochester, MN, providing the rapid identification of patients who could benefit from palliative care.  

Data Scientists at the Mayo Clinic have experienced several instances of model degradation over time with the palliative care model, which they have attributed to a variety of changes in the underlying data stream.  To address this need for near real-time model monitoring, we create a general framework, called \textit{Bayes Watch}, for detecting change-points in real-world, messy steaming data and assessing what features in the data led to this change.  Finding changes in the input data is a natural means to drive model monitoring decisions; if one were to only rely on analyzing model outcomes, a significant change in the model may not be detected until well after the moment of model degradation, since collecting enough outcome data may take weeks or even months.  While the Bayes Watch approach is well-suited to the specific problem of monitoring prediction models in health care settings, the method is a general approach to a wide variety of change-point problems.  Unlike other change-point approaches, Bayes Watch leverages modern advances in learning graphical models, handles a wide array of data challenges, is able to simultaneously check for many different types of change-points, and comes with a unique fault-detection system that is shown to out-perform industry standards. 

To perform fault detection, Bayes Watch keeps a sparse log of Gaussian Graphical Mixture Models (GGMMs), which are finite mixture models where each component is a Gaussian Graphical Model (GGM) \citep{everitt1996, uhler2018}.  This log is kept throughout the MCMC sampling used to fit the change-point model.  After determining a change-point, these logs are used to analyze the differences between the data distributions before and after the change-point.  We introduce two novel metrics for this analysis, which are inspired by the variance decompositions used to calculate Sobol' indices \citep{sobol2001}.  These metrics are shown to perform well on simulated data, and on two separate time periods of the data fed into the \cite{murphree2021} palliative care model.

Beyond its flexibility against numerous data challenges, we show that the change-point model within Bayes Watch is able to catch a wide variety of change-point types.  In addition to changes in the center and spread of the data distribution over time, this system identifies nuanced changes in the way that data are missing, which is of special interest to data engineers trying to diagnosis data pipeline issues.  For data with values missing at random, Bayes Watch identifies significant changes in the amount of missing values.  For data with values missing according to some inter-related structure (data not missing at random) which is common in health care data, Bayes Watch is able to detect changes in this structure over time.  As far as the authors are aware, this is the first time in the literature that changes in the inter-related structure of missingness have been considered. 

The remainder of this paper is laid out as follows.  We begin with a comprehensive review of the existing change-point literature, and a longer discussion of the \cite{murphree2021} model's data and its challenges.  The technical details of the Bayes Watch model are outlined in Section \ref{sec:method}, including details on an MCMC algorithm developed for sampling regimes and fitting the model parameters.  This section also introduces a split-merge-swap algorithm as an efficient sampling scheme for the regimes.  The model is investigated using a comprehensive array of simulated data in Section \ref{sec:sim}.  This section provides evidence for the versatility of this model by simulating many different types of possible change-points under the influence of several data challenges.  In Section \ref{sec:posterior}, we develop a novel fault-detection technique for identifying what features led to a change-point after determining that one exists.  Our analysis of the Mayo Clinic palliative care model using Bayes Watch is provided in Section \ref{sec:app}.
A user-friendly version of the Bayes Watch framework is available as the R package `bayesWatch' on CRAN \citep{r2021}.  Questions and suggestions on this package can be directed to the corresponding author.

\subsection{Relevant Literature}
Statistical estimation of change-points dates back to the early 1900s, when researchers wanted to locate abrupt changes in the center of independent and identically distributed (iid) data for quality control in industry \citep{shewhart1931, page1954}.  Since these early developments, change-point detection has remained a regular topic for researchers, owing in part to the wide breadth of fields in which the problem arises \citep{truong2020}.  While there is a huge number of change-point papers and a wealth textbooks on the topic \citep{basseville1993, brodsky1994,chen2014}, the field still sees plenty of modern advancements.  This paper addresses the problem of detecting multiple change points in multivariate time series, so we limit the following literature review to approaches that specifically handle this problem.

Likely one of the most broadly used approaches to change-point detection is the Binary Segmentation (BS) algorithm.  In the BS algorithm, the data are searched for a single change-point that is then used to divide the data into two subsets.  This process is repeated recursively on each new set of subsets until a stopping criterion is achieved \citep{vostrikova1981}.  A classic BS algorithm has several drawbacks.  First, since the BS algorithm is essentially an optimization problem, it is unable to fully capture the uncertainty in the number and location of the change-points \citep{nam2012}.  Second, the BS algorithm is not consistent unless some amount of minimum spacing is assumed between change-points.  This is proven in \cite{fryzlewicz2014}, who give a potential solution, but only for the one-dimensional case.  Furthermore, we have not found many instances of the BS algorithm being used on high-dimensional data ($p > 100$), nor on data with missing values and mixed data types.  The paper by \cite{cho2015} introduces a variant of the BS algorithm that addresses the high-dimensional case, but only up to $p=100$.  The issue of change-point detection on data with missing values is addressed in \cite{londschien2021}, who introduce a framework that integrates the BS algorithm with data imputation schemes.  Once change-points are located, \cite{londschien2021} fit GGMs to regimes using the Graphical LASSO \citep{friedman2008}.  Like the previous BS approaches, \cite{londschien2021}'s method suffers from a lack of uncertainty quantification on the individual change-points, and is only tested on dimensions up to $p=100$. 



Besides the BS algorithm, there are other optimization schemes for determining the number and location of change-points in multivariate timeseries.  In \cite{lavielle2007}, finding multiple change-points is framed as a global optimization problem over the space of all possible change-point configurations (in contrast to the BS algorithm, which finds change-points locally).  The estimator on the location and number of change-points is calculated efficiently using a dynamic programming scheme.  In \cite{ombao2005}, multivariate timeseries are analysed using the SLEX (smooth localized complex exponentials) basis, which is a basis that is localized in time.  Choosing a proper SLEX model requires that the timeseries be segmented into approximately stationary segments, which \cite{ombao2005} do using a complexity-penalized log energy optimization problem.  While these approaches do handle the multivariate change-point problem, the performance of these approaches is only tested for small dimension sizes ($p<20$).  Furthermore, neither come with robust uncertainty quantification schemes nor directly handle the data challenges, such as mixed variable types and missing values, considered in this paper.

Change-point detection has seen significant treatment in the Bayesian literature, with applications in climate data \citep{ruggieri2013}, DNA segmentation \citep{fearnhead2007}, and finance data \citep{chopin2007}.  As far as the authors are aware, the earliest Bayesian approach for location change-points was in \cite{chernoff1964}, although the primary aim of this paper was not change-point detection but the estimation of the final mean of a sequence of iid Gaussian data without knowing when the means change.  Numerous other methods were formulated through the early 2000s.  In \cite{koop2007}, multiple change-points are located without knowing the number of change-points a priori by assuming that regime duration follows a Poisson distribution.  The method introduced by \cite{martinez2014} takes a nonparametric approach using a Poisson-Dirichlet process prior \citep{buntine2012}.  A popular formulation, which we consider in detail in this paper, is the approach initially described in \cite{chib1998}.  In \cite{chib1998}, change-points are modeled by learning a discrete random variable that indicates the regime from which a particular observation is drawn.  This method has seen numerous extensions and modern updates, such as in \cite{ko2015}, where a Dirichlet process prior is used to define a hidden Markov model with a potentially infinite number of states.  This removes the need to specify the number of states a priori.  Missing data is addressed in \cite{corradin2022}, who add a data imputation step in the sampling scheme introduced by \cite{martinez2014}. However, this method is only tested for dimension sizes up to $p=4$.

While there are many change-point methods and applications in the literature, there is not one readily available that handles every data challenge present in data like those considered in \cite{murphree2021} (see Section \ref{subsec:data}).  The contributions of this paper include the design of a general change-point model that directly addresses data with mixed types (continuous, binary, ordinal, and nominal data), missing values, and high dimensions ($p>100$), while also quantifying the uncertainty of each prediction using a Bayesian framework.  Bayes Watch achieves these aims by classifying data over time into different regimes of homogeneous data according to a Markov Process with unknown transition probabilities, where each regime is fit using a GGMM.  This leads to a posterior distribution on a vector of regime classifications, which is used to provide meaningful expressions on the probability of every possible change-point.  GGMMs are used because data in the wild, such as the data examined in this paper, are not well approximated by the normal distribution, so a normal mixture is needed.  However, the performance of Gaussian models and particularly Gaussian mixtures in high dimensions is known to suffer due to the $p^2$ growth of the number of covariance parameters.  The proposed remedy is a sparse precision approximation to the normal components to limit the number of parameters and preserve degrees of freedom for the estimation/identification of multiple regimes.  The graph structure that determines the sparsity of the precision matrices for the GGMMs is given a uniform prior distribution and learned along with the model parameters and the unknown transition probabilities.  The sampling of the graph structure is done using the Double Reversible Jump algorithm of \cite{lenkoski2013}.  Under a Bayesian hierarchical framework, hyperpriors are assigned to each prior parameter, which is discussed in later sections.

\subsection{Challenges for Process Monitoring of Hospital Data}\label{subsec:data}

While events that cause a change in model performance can be detected via monitoring outcomes, the outcomes can be somewhat rare and weeks or even months can pass before enough outcomes are observed to determine that model performance has degraded.  This issue is somewhat exacerbated in healthcare, since waiting until stark changes are eventually observed in model performance could mean a failure to address a patient's medical need and leads to care providers loosing trust in the model.  Thus, if there is an event that causes model degradation, this would need to be addressed as soon as possible and corrected, ideally before changes in model performance are observed.  Significant changes in the underlying input (or feature) data are strong proxies for performance failures in the model, since the data used to tune the model would no longer represent the current reality.  There are many ways that this underlying data might change, many of which have already been experienced by data scientists working on the palliative care model from \cite{murphree2021}.  In one example, the units of a laboratory test in the electronic medical record (EMR) changed without the knowledge of the modeling team.  In another example, a data pipeline issue led to a significant increase in missing data values for certain variables.    


We consider two segments of data fed into the Murphree model: one during February of 2020 and one from mid-May to mid-June of 2020.  These segments represent data from before and after the start of the broad shutdowns in the US due to the COVID-19 pandemic.  Within these datasets, there are 123 total variables that consist of 35 continuous, 87 binary, and 1 nominal, with 27 variables having some amount missing values.  Within each discrete variable in this application, there are no missing data values.  During the February segment, there were 233,342 data observations of a total of 13,542 patients.  During the mid-May to mid-June segment, there were 298,039 data observations of a total of 8,060 patients.  
As discussed in Section \ref{subsec:discrete}, additional dummy variables are added to the data to encode the locations of missing values prior to the data imputation step, and some variables were removed due to zero variance in these time periods.  With these changes, the total number of variables in the data is 182.

Each data object in this paper is a single observation on a patient in time that is fed into a predictive model call for a palliative care prediction.  These observations occur whenever a patient is moved to a different unit within the hospital or when a billable procedure/test occurs such as a laboratory test, documented diagnosis, etc.  Thus, there are therefore multiple observations for each patient.  For the purpose of this application, these multiple observations are assumed to be iid.  However, it is a topic of further work to examine the benefits of treating the unit of observation as the entire longitudinal patient encounter.  The data stream has a fixed set of variables that encode a broad set of information; they include patient labs, the presence/absence of specific chronic diseases, basic demographics, and information about a patient's movement within the hospital.  The data are grouped by the days that they occur.  A breakdown of the variables in this data, and their data types, is available in the Supplementary Material. 

The palliative care model prediction data pose many data challenges from a modeling perspective.  There is a wide variation in the proportion of missing values for the continuous variables; the average missing proportion within the continuous variables is 0.215, with a standard deviation of 0.282.  Furthermore, the continuous variables were mostly non-normal.  Although non-normality can be handled within the model monitoring framework proposed here, it does come with an added computational cost of having to learn many mixture components in the mixture model.  To limit the number of components needed in the mixture model and to ease the computational burden, the continuous variables are pre-processed using the Box-Cox transform, optimizing the transformation parameter according to the method by \cite{guerrero1993}.

Due to Health Insurance Portability and Accountability Act (HIPAA) considerations, the real data used in this paper is not available for release.  However, simulated data that mimics this real data is available within the `bayesWatch' package.

\section{Methodology}\label{sec:method}

\subsection{Missing, Discrete Variables, and Boundaries/Censoring}\label{subsec:discrete}
Let $\mathbf{Y} = (\mathbf{y}_1, \dots, \mathbf{y}_T )'$ be a timeseries of $I_t \times J$ matrices of data $\mathbf{y}_t$: $t \in \{1, \dots T\}$.  The matrix $\mathbf{y}_t$ refers to the data observed at the $t$th timepoint, which we assume consists of $I_t$ independent observations of a $J$-dimensional random vector.  Let $y_{t;ij}$ denote the $j$th variate of the $i$th observation at the $t$th timepoint ($i \in \{ 1, \dots, I_t\}$ and $j \in \{1, \dots, J\}$).  This setup is meant to mimic data streams of a fixed set of variables grouped by the days that they occur.  

To handle cases where the data $\mathbf{Y}$ are binary, ordinal, nominal, censored, or missing, a Gaussian latent variable approach is used, which encodes the observed data in terms of unknown random values.  In the following, we outline the technical details of this approach. 

Suppose the sets $\mathcal{D}, \mathcal{D}^c$ form a partition of the set $\{1, \dots, J\}$ where the variables $y_{t;ij}$ are discrete whenever $j \in \mathcal{D}$ and continuous whenever $j \in \mathcal{D}^c$, for every observed timepoint $t.$  For every $y_{t;ij}$ such that $j \in \mathcal{D}^c$, let $\infty \leq b_j \leq \mathbf{y}_{t; i, j} \leq c_j \leq \infty$ be the upper and lower limits for all observations, for all $t$ and $i$.  For variables without left or right boundaries on their values, set $b_j = -\infty$ and $c_j = \infty$.  For the time being, assume that all discrete values are ordinal, where variable $j\in \mathcal{D}$ has $L_j$ levels $d_{j,1}, \dots, d_{j,L_j}$.  Define the latent variables $\mathbf{Z} = ( \mathbf{z}_1, \dots, \mathbf{z}_T )'$ as a series of $I_t \times J$ matrices, parallel to $\mathbf{Y}$, such that all observations are continuous, $J$-variate random vectors where
\begin{equation}\label{eq:latent_setup}
    y_{t;ij} = \begin{cases}
\sum_{l=1}^L d_{j,l}\1(a_{j,l-1} < z_{t;ij} \leq a_{j,l}) & \text{for }j \in \mathcal{D} \\
z_{t;ij}\1( b_j \leq z_{t;ij} \leq c_j) + b_j \1(z_{t;ij} < b_j) + c_j \1(z_{t;ij} \geq c_j)  & \text{for }j \in \mathcal{D}^c
\end{cases},
\end{equation} 
similar to the setup found in \cite{storlie2018}.  The cutoffs $a_{j,l} := d_{j,l}$ for $l = 1, \dots, L_j - 1$ and $a_{j,0} =-\infty, a_{j,L_j} =\infty.$  This approach accounts for binary variables ($L_j = 2$), where $a_{j,1}$ and $d_{j,1}$ are each set to 1.  Nominal variables can be modeled according to a latent variable approach similar to the set-up in \eqref{eq:latent_setup}.  The underlying latent structure for a nominal variable $y_{t;ij}$ with $L_j$ levels is a vector of values $\mathbf{z}_{t;ij} = (z_{t;ij1}, \dots, z_{t;ijL_j})$, where $y_{t;ij} = d_{j,l}$ if $l = \argmax_l z_{t;ijl}$ \citep{bhattacharya2012b}.  Note that this approach will cause the elements of $\mathbf{Z}$ to no longer be $J$-dimensional.  However, since the data $\mathbf{Y}$ are only used to update the latent data $\mathbf{Z}$, this does not come with any major computational difficulties.

Missing values are handled by imputing samples from the multivariate normal distribution used to model the latent data, conditioned on the values that are not missing.  This process is outlined in the Supplementary Material, and is similar to the approach outlined by \cite{corradin2022}.  While imputing is very useful for fitting the overall model, it removes information on the way that the data missing.  According to the experiences of researchers at the Mayo Clinic, this information is very useful, since it may indicate a data pipeline issue that should be caught by the model.  To keep this information, indicator variables are added to the data that encode instances of missing data.  That is, for a variable $\mathsf{foo}$ that has some number of missing values, create the additional binary variable $\mathsf{fooNA}$ in the data where
\begin{equation}\label{eq:missing_indicators}
    \mathsf{fooNA}_i = \begin{cases}
1 & \text{if $\mathsf{foo}_i$ is missing} \\
0 & \text{if $\mathsf{foo}_i$ is not missing}
\end{cases}.
\end{equation} 

\subsection{Markov Process Change-point Model}\label{subsec:HMM}

To determine change-points, we will use the discrete change-point model setup from \cite{chib1998}, which is similar to that of \cite{ko2015}.  In the following, we provide some technical details on this model.

Let $R, T, k \in \nn$.  Each observed timepoint, which can be written in terms of a matrix of latent variables $\mathbf{z}_t$ according to \eqref{eq:latent_setup}, is assumed to come from a distribution depending on parameter sets from the collection $\boldsymbol\Theta := \{ \boldsymbol\theta^{(1)}, \dots, \boldsymbol\theta^{(R)}\}$ whose value changes at unknown change-points $\boldsymbol\tau := \{ \tau_1, \dots, \tau_k \}$ such that $1 < \tau_1 < \dots < \tau_k < T$.  The change-points $\boldsymbol\tau$ segment all discrete timepoints $\{ 1, \dots, T\}$ into groups called \textit{regimes}.  Assume that $\boldsymbol\theta^{(r)}$ is constant within a regime $r$, but differs across regimes.  Notationally,
\begin{equation}\label{eq:regimeDistns}
  \mathbf{z}_t \sim \begin{cases} 
      p(\mathbf{z}_t | \boldsymbol\theta^{(1)}) & \text{if } t \leq \tau_1, \\
      p(\mathbf{z}_t | \boldsymbol\theta^{(2)}) & \text{if } \tau_1 < t \leq \tau_2, \\
      \vdots & \vdots \\
      p(\mathbf{z}_t | \boldsymbol\theta^{(k+1)}) & \text{if } \tau_k < t \leq T.
   \end{cases} 
\end{equation} 
The number of possible regimes (distinguished by the superscript $r$) is distinct from the number of change-points (distinguished by the subscript $k$).  In general, superscripts are reserved in this paper to denote regimes.  When $k < R$, the parameter sets of regimes without data are learned according to their priors (discussed in Section \ref{subsec:dists}).  Meanwhile, $R$ should be chosen to be sufficiently large to make the case of $k > R$ impractical.

Let $\boldsymbol\phi := (\phi_1, \dots, \phi_T )$ be a discrete state vector such that the value of $\phi_t$ denotes the regime to which $\mathbf{z}_t$ belongs:
\[ \mathbf{z}_t | \phi_t \sim p(\mathbf{z}_t | \boldsymbol\theta^{(\phi_t)} ). \]
Since $\boldsymbol\phi$ represents regime membership of observations over time, the $\phi_t$s must be sequentially ordered, with $\phi_{t+1} \in \{ \phi_t, \phi_t + 1 \}$.  Note that this prevents the model from returning to a previously learned regime, which is done for computational simplicity in the MCMC sampling scheme.  Allowing for a regime change that returns to a previously observed regime would be an interesting aim for future work.  Learning the state vector $\boldsymbol\phi$ is the primary inferential concern, since $\boldsymbol\phi$ determines the change-points in the model completely.

The state vector $\boldsymbol\phi$ is modeled as the observed states of some Markov process constrained to mimic the forward movement of time.  Fix the number of possible regimes $R \in \nn$.  The transition matrix of such a process has the form
\begin{equation*}
P = \begin{pmatrix}
\mathsf{p}^{(1)} & 1 - \mathsf{p}^{(1)} & 0 & \dots & 0 & 0 \\
0 & \mathsf{p}^{(2)} & 1-\mathsf{p}^{(2)} & \dots & 0 & 0 \\
\vdots & \vdots & \vdots & \ddots & \vdots & \vdots \\
0 & 0 & 0 & \dots & \mathsf{p}^{(R - 1)} & 1 - \mathsf{p}^{(R - 1)} \\
0 & 0 & 0  & \dots & 0 & 1 \\
\end{pmatrix},
\end{equation*}
where $1 - \mathsf{p}^{(r)}$ is the probability of moving from regime $r$ to regime $r+1$: 
\begin{equation*} 
p(\phi_t = r+1 | \phi_{t-1} = r) = 1 - \mathsf{p}^{(r)} \text{ for } r \in \{ 1, \dots, R - 1\}.
\end{equation*}
For this setup, it is not possible to return to a previous regime.  In practice, the value $R$ is set to be some large integer well above the number of regimes one expects to observe.  For applications with a fixed number of possible locations for a change $T$ (in our case this would be the number of distinct days), one can exhaust all possibilities by setting $R = T$.

\subsection{Dirichlet Process Mixture Model}\label{subsec:dirichlet}
In many applications, it may be unreasonable to assume that the data within a given regime follows a multivariate normal distribution.  While data transformations can help mitigate non-normality, these approaches are largely ineffective on multi-modal data.  Using the Dirichlet process (DP) prior, a mixture distribution is able to take a potentially infinite number of components, which admits a flexible distribution that can fit even starkly non-normal, multi-modal data.  The model parameters within a regime, $\boldsymbol\theta^{(r)}$ for $r \in \{1,\dots,R\}$, therefore follow the model
\begin{equation}\label{eq:dirichlet}
    \mathbf{z}_{t; i,[J]} | \boldsymbol\theta^{(r)} \sim F(\boldsymbol\theta^{(r)}),~~~~~\boldsymbol\theta^{(r)} \sim H,~~~~~H \sim DP(H_0, \alpha),
\end{equation} 
where $H \sim DP(H_0,\alpha)$ means that $H$ is a random distribution generated by a DP with base measure $H_0$, $\alpha$ is the precision parameter \citep{maceachern1998}, and $[n]:=\{1,\dots,n\}$ for any integer $n$.  In this general set-up, the likelihood function $F$ could be taken to be any distribution parameterized by $\boldsymbol\theta^{(r)}$; this is taken to be a Gaussian for our applications. 

Let $\boldsymbol\gamma_t := \{\gamma_{t;1},\dots \gamma_{t;I_t}\}$ be the component assignments for timepoint $t$ generated by the Dirichlet process.  The posterior distribution for $\boldsymbol\gamma_t$ is often calculated by integrating over $H$.  However, as a consequence of the design of the Gibbs sampler proposed in Section \ref{sec:mcmc}, its distribution is instead approximated using a truncated stick-breaking process.  Let $Q$ be some large integer that we use to truncate the DP, and let $\Pi := \{\pi_1, \dots, \pi_Q\}$ be the probability of an observation belonging to components $\{1,\dots, Q\}$.  The posterior probability that $\gamma_{t;i}$ belongs to component $k \in \{1,\dots,Q\}$ is
\[ p(\gamma_{t;i} = k) \propto \pi_k f(z_{t;i} | \theta_{t;k}),~~~~~ v_k | \boldsymbol\gamma_t \backslash \gamma_{t;i} \sim \text{Beta}\left( 1 + \sum_{j=1}^{I_t} \1(\gamma_{t,j} = k), \alpha + \sum_{j=1}^{I_t} \1(\gamma_{t;j} > k) \right), \]
where the sticks $v_k$ map to the component probabilities,
\[ \pi_k = v_k \prod_{j=1}^{k-1} (1 - v_j), \]
and $\1(\cdot)$ denotes the indicator function.

In the simulation study, the authors found good model fits with a relatively low number of components (between 1 and 3, usually), so $\alpha$ was set sufficiently small and a $Q$ of 7 was chosen to avoid the DP from ever being actively truncated.  While more components would mean a closer fit to the observed data, this comes with an increased computational burden.  The balance between a sufficiently close fit of the model and the computations required to achieve this fit is controlled by $\alpha$, and should be approached on a case-by-case basis.

\subsection{Model and Prior Distributions}\label{subsec:dists}
The Bayes Watch framework requires the simultaneous sampling of many nuisance parameters alongside the sampling of the parameters of primary interest: $\boldsymbol\Theta$ and $\boldsymbol\phi$.  We perform this fitting using a Bayesian hierarchical framework, which allows for efficient learning by way of Gibbs sampling with multiple Metropolis-Hastings updating steps.  In the following, we outline the distributional assumptions on the data, provide a comprehensive list of parameters that need to be sampled for this model, and state the prior distributions for each.

The elements of the matrix $P$ are modeled using Beta$(w,v)$ priors on each $\mathsf{p}^{(i)}$.  This is notably different to the common approach of using a Dirichlet Process prior with an infinite number of states, which allows one to integrate out the $\mathsf{p}^{(i)}$'s, removing the need for learning them \citep{beal2002}.  However, the Dirichlet Process approach encodes into the framework a prior belief that the regimes are ``sticky," i.e. that the larger a regime is, the less likely it is for a regime to change.  This is notably not the case for the Mayo Clinic palliative care model data, and in general not an expectation that the authors have for this model.  Thus, it is necessary to learn these parameters.  Embracing a Bayesian hierarchical framework, Gamma$(a,b)$ priors are assigned to the hyperparameters $\alpha$ and $\beta$, which are each learned separately in the final MCMC chain. 

Since the elements of the matrix $P$ are learned explicitly, the prior distribution of the state vector $\boldsymbol\phi$ is known.  The vector $\boldsymbol\phi$ is updated according to a merge-split-swap approach similar to the one proposed in \cite{martinez2014}, but with a more data-driven approach to sampling a split point (See Section \ref{sec:mcmc}).  

The distribution $F$ from \eqref{eq:dirichlet} is taken to be a GGMM with parameters $\boldsymbol\theta^{(r)}:=\{ \boldsymbol\mu^{(r)}, \boldsymbol\Lambda^{(r)}, G\}$, where $\boldsymbol\mu^{(r)}, \boldsymbol\Lambda^{(r)}$ are collections of mean vectors and precision matrices for all possible $Q$ components in regime $r$.  This model is chosen to avoid imposing a normality assumption on real-world data, while also providing a sparsity assumption to help with an exploding $p^2$ growth in the number of parameters that need to be learned.  Let $G$ be some undirected graph $G := (V, E)$ for $V := \{1, \dots, J\}$ and $(q,s) \in E$ if and only if the edge between vertices $q,s \in V$ is represented in $G$, and let $\Lambda^{(r)}_{\gamma_{t;i}}, \mu^{(r)}_{\gamma_{t;i}}$ be the precision matrix and mean vector, respectively, of the $\gamma_{t;i}^{\text{th}}$ component of regime $r$.  Following typical conventions for GGMs, the zeros of a precision matrix $\Lambda^{(r)}_{\gamma_{t;i}}$ align with the missing edges of the graph structure $G$: $(\Lambda^{(r)}_{\gamma_{t,i}})_{q,s} = 0$ if and only if $(q,s) \notin E$, where $(\mathsf{c})_{\mathsf{a},\mathsf{b}}$ means the submatrix with only the row(s) $\mathsf{a}$ and the column(s) $\mathsf{b}$ of the matrix $\mathsf{c}$.  Thus, given the component assignment vector $\boldsymbol\gamma$, the likelihood for the $i^{\text{th}}$ observation in regime $r$ is $\text{MVN}(\mu^{(r)}_{\gamma_{t;i} }, \Lambda^{(r)}_{\gamma_{t;i} }, G)$ whenever $\tau_{r-1} < t \leq \tau_r$, where $\text{MVN}$ refers to the mixture multivariate normal distribution with fixed graph structure. While $G$ is learned along with the parameters $\mu^{(r)}_{\gamma_{t;i}}$ and $\Lambda^{(r)}_{\gamma_{t;i}}$, we make the practical assumption that $G$ is the same across components and across regimes.

The base distribution $H_0$ in \eqref{eq:dirichlet} is taken to be a Normal $G$-Wishart (NG-W), which we introduce as the generalization of the Normal Wishart prior given the constraint imposed by a known graph structure $G$.  This is also therefore a generalization of the normal-normal Mixture Dirichlet Process for when a graph structure $G$ is known \citep{maceachern1998}.  As we will show, the NG-W is the conjugate prior of a multivariate normal distribution with unknown parameters $(\mu^{(r)}_{\gamma_{t;i}}, \Lambda^{(r)}_{\gamma_{t;i}})$, but known graph $G$.  Define the NG-W analogously to how the Normal Wishart is defined as in \cite{bishop2006}: as a joint distribution of $(\mu^{(r)}_{\gamma_{t;i}}, \Lambda^{(r)}_{\gamma_{t;i}})$, determined by the conditionals
\begin{equation*} 
\mu_{\gamma_{t;i}}^{(r)}| m, \lambda, \Lambda_{\gamma_{t;i}}^{(r)} \sim \mathcal{N}(m, (\lambda \Lambda_{\gamma_{t;i}}^{(r)})^{-1}), ~~~~~\Lambda_{\gamma_{t;i}}^{(r)} | D, \nu \sim \mathcal{W}_G (\Lambda | D, \nu).
\end{equation*}
Here, $m$ and $(\lambda \Lambda_{\gamma_{t;i}}^{(r)})^{-1}$ are the mean vector and covariance matrix of a multivariate normal distribution, while $D$ and $\nu$ are the scale matrix and degrees of freedom of a $G$-Wishart distribution, using the formulation defined by \cite{wang2012}:
\begin{equation*}  
f_{\mathcal{W}_G(\nu, D)}(\Lambda) = I_G(\nu, D)^{-1} |\Lambda|^{\frac{\nu-2}{2}} \exp \left\{ - \frac{1}{2} \tr\left( D' \Lambda \right) \right\}\1_{\Lambda \in \mathcal{P}_G}, 
\end{equation*}
where $\mathcal{P}_G$ is the space of positive-definite, symmetric precision matrices with a zero structure that matches the graph $G$, and $f_X$ denotes the PDF of the random variable $X$.  The term $I_G(\nu, D)$, called the $G$-Wishart normalizing constant, is defined as:
\begin{equation*} 
I_G(\nu, D) = \int |\Lambda|^{\frac{\nu-2}{2}} \exp \left\{ - \frac{1}{2} \tr\left( D' \Lambda \right) \right\}\1_{\Lambda \in \mathcal{P}_G}d\Lambda. 
\end{equation*}
Thus, $(\mu_{\gamma_{t;i}}^{(r)}, \Lambda_{\gamma_{t;i}}^{(r)}) \sim \text{N-GW}(m, \lambda, D, \nu),$ with PDF
\begin{equation} \label{eq:NGW}
    f_{(\mu, \Lambda)} = f_{\mathcal{N}(m, (\lambda \Lambda)^{-1})} f_{\mathcal{W}_G(D, \nu)}.
\end{equation}

The NG-W is a well-defined probability distribution, and inherits the conjugacy of the Normal Wishart due to the following proposition:
\begin{proposition}
Let $\mathcal{P}$ be some family of probability distribution functions and $y \sim f(y | \theta)$ be some observed data where $\theta \sim f(\theta) \in \mathcal{P}$ implies $f(\theta | y) \in \mathcal{P}$.  Assume that there exists some set $A$ of the parameter space such that $\int_A f(\theta | y) d\theta > 0$ and $\int_A f(\theta) d\theta > 0$.  Let $\mathcal{P}_A$ be another family of probability distributions, defined such that $f(\theta) \in \mathcal{P}$ implies $f^\star(\theta):=  c\cdot f(\theta) \1_A \in \mathcal{P}_A$, where $c$ is the normalizing constant so that $f^\star(\theta)$ is a density.  Then $\mathcal{P}_A$ is a conjugate family.
\end{proposition}
\begin{proof}
$f^\star(\theta |y) = \frac{f^\star(\theta) f(y | \theta)}{\int f^\star(\theta) f(y | \theta) d\theta} = \frac{c \cdot f(\theta) f(y | \theta)\1_A}{\int f(\theta) f(y | \theta) d\theta} \cdot \frac{\int f(\theta) f(y | \theta)d\theta }{ \int f^\star(\theta) f(y | \theta)d\theta} \propto f(\theta|y) \1_A \in \mathcal{P}_A.$
\end{proof}

The following Corollary is an immediate consequence of the above result and the fact that the Normal-Wishart Distribution is conjugate.  
\begin{corollary}\label{cor:conjugate}
The NG-W distribution, as defined in Equation \ref{eq:NGW}, is a conjugate prior for MVN data, assuming that $\int_{\rr^J } \int_{\mathcal{P}_G } f_{\mathcal{W}_G(D,\nu)}(\theta) f_{\mathcal{N}(\mu, \Lambda)} (y) d\mu ~ d\Lambda > 0$.
\end{corollary}
The hyperparameter $m$ is learned using a Gamma conjugate prior; the hyperparameter $\lambda$ is learned using a Normal conjugate prior.  The $D$ scale matrix is learned using a Wishart prior.

The prior for $G$ is determined by the probability that an edge is present, assuming that edge representation is independent over all possible edges.  Let $\rho$ be the probability that a given edge is present.  Then,
\begin{equation*} 
P(G) = \rho^{|G|}(1- \rho)^{ J(J-1)/2 - |G| }, 
\end{equation*}
where $|G|$ is the number of unique elements in $G$.
Alternative choices for the prior distribution on the graph space $\mathcal{G}$ are plentiful, such as priors that encourage sparse graph structures \citep{jones2005} and those based on multivariate discrete distributions \citep{scutari2013}. 

\subsection{Point Estimate of \texorpdfstring{$\boldsymbol\phi$}{}}\label{subsec:pointestimate} The MCMC scheme described here simulates an empirical posterior distribution for the regime vector $\boldsymbol\phi$.  This distribution is a rich expression of the space of feasible regime vectors used to attribute probabilities to the space of possible change-points.  Given samples on $\boldsymbol\phi$, one can calculate an estimate for the change-point probability vector,
\[ \mathfrak{C} = (\mathfrak{c}_1, \dots, \mathfrak{c}_{T-1}), \]
where $\mathfrak{c}_t$ is the probability that a change-point occurs after day $t$.  Each value $\mathfrak{c}_t$ is estimated by the proportion of $\boldsymbol\phi$ samples that see a regime change directly after day $t$.  Using $\mathfrak{C}$, it is straightforward to derive a point estimate of the true vector of change-points; given some cutoff value $\mathfrak{c}^\star$, determine a change-point after every day $t$ for which $\mathfrak{c}_t \geq \mathfrak{c}^\star$.

The choice for the cutoff value $\mathfrak{c}^\star$ is ideally chosen based on user experience in the context of a given data application, with the practical aim of controlling for the false positive rate and thus limiting alarm fatigue.  In use cases where there is no context for a reasonable choice of $\mathfrak{c}^\star$, the authors suggest the following procedure.  Simulate draws for $\boldsymbol\phi$ and $\boldsymbol\Theta$ and let $\hat{\boldsymbol\phi}$ be the maximum \textit{a posteriori} (MAP) estimate.  Treating $\hat{\boldsymbol\phi}$ as the true regime vector, simulate multiple datasets using $\boldsymbol\Theta$ and run the MCMC algorithm on each of these data sets.  Using these secondary posterior draws, set $\mathfrak{c}^\star$ to a value that gives a desired false positive rate against the initial estimate of true regime vectors $\hat{\boldsymbol\phi}$.

One drawback of the above approach is the added computational burden, which becomes particularly cumbersome for applications to data with more than 200 variables.  Through many experiments, we found a cutoff of 0.5 to be a reasonable default choice in our investigations.  This being said, a researcher should still look at the entire spread of change-point probabilities when analyzing the model output.

\subsection{MCMC Algorithm}\label{sec:mcmc}
Under the Gibbs framework, new values are simulated for each parameter by first fixing the value of every other parameter.  The specific distributional details for every prior and sampling distribution used to fit the change-point model are technical in nature and are available in the Supplementary Material.  The MCMC chain samples the parameters in the following order: draw latent data $\mathbf{Z}$; draw hyperparameters $D, \boldsymbol\gamma, P, w, v, m, \lambda, \nu$; draw regime vector $\boldsymbol\phi$ (merge-split); draw regime vector $\boldsymbol\phi$ (swap); draw parameter collections $\Theta$; draw graph structure $G$.  Initial values are chosen from the prior distributions or from empirical estimates using the data, when these are available.  The Bayes Watch framework uses the Double Reversible Jump (DRJ) algorithm of \cite{lenkoski2013} to efficiently resample the graph structure.  The DRJ combines an exact sampler of the $G$-Wishart distribution, the RJ procedure in \cite{dobra2011}, and the exchange algorithm of \cite{wang2012}, to create a MCMC procedure that sidesteps the need to calculate the $G$-Wishart normalizing constant, since this computation can be a computational bottleneck.  While the DRJ algorithm is a useful way to avoid calculating the normalizing constant of the $G$-Wishart distribution, there remains a need to occasionally calculate this constant within Bayes Watch.  This need arises when resampling the regime vector $\boldsymbol\phi$, since changing the underlying model causes a mismatch in the Metropolis-Hastings ratio regardless of our choice of a conjugate prior.  Details on this mismatch and further details on the DRJ algorithm are available in the Supplementary Material.

\section{Model Validation}\label{sec:sim}

We use a simulation study to explore the robustness of the Bayes Watch change-point model against many different types of regime changes.  For this study, continuous data are generated from a multivariate normal model, and variables are transformed to discrete and/or missing values are imputed as needed.  Due to the wide variation of possible change-point types, this study will examine many forms of regime changes, as outlined in Table~\ref{tab:sim_types}.  The forms for the change-points in Table~\ref{tab:sim_types} are subtle and meant to challenge each approach's ability to detect different types of changes.  These change-points are intentionally made more difficult to detect by considering additional data challenges, such as bimodal data, mixed data types, and missing values.  The bimodal mean change example is of particular note; it creates a bimodal distribution where the overall data mean and variance across regimes remains constant, yet the centers of each mode of the data distribution are shifted in equal and opposite directions.  

As a final challenge, data are created where the center and spread across regimes are the same, but the amount of and pattern between the missing values changes.  These missing data challenges are created by modifying the missingness indicator variables from \eqref{eq:missing_indicators} between regimes.  A change in the amount of missing values is simulated by imposing a mean change on these indicators, while a structural change is achieved by resampling the submatrix of the overall covariance matrix that corresponds to the missing indicators only\footnote{This resampling is done according to an inverse Wishart distribution \citep{storlie2018}.}.  

We compare the change-point model within Bayes Watch against industry standards: the Hotelling $T^2$ Scan Statistic (HT2) \citep{charts2012} and the Real-Time Contrasts (RTC) Method \citep{deng2012}.  In the RTC method, a classifier is fit to distinguish between a fixed window of days and a single new day (or days).  If the accuracy of this classifier is sufficiently high on a permutation of the training data, the method determines that a regime change has occurred.  For the purposes of this data simulation, a GBM is chosen as the classifier due to its robustness and ability to inherently handle missing data.  

These moving window statistics are given a memory of three days; testing if, at each new day, the data are significantly different than the previous three days.  The need to choose a window size is a drawback of these scan-type approaches; here, we chose a 1-day window due to the need to be able to detect changes the day of the change as opposed to several days later.  The authors found that the HT2 performs very well for detecting mean changes for continuous, unimodal normal data, but performs poorly for different types of change-points.  The RTC method performs poorly for all test cases.  The poor performance of the RTC method is largely due to the subtle signals used in the simulations.  For much starker regime changes, the RTC Method's performance (perhaps obviously) improves.  

Two variations of the Bayes Watch change-point model are included in the simulation: one that fits a regular GGM to each regime, and another that allows for non-normal data by fitting the full Mixture GGM.  The computation time for the Mixture GGMs is significantly longer than the regular GGMs, but this model is more general and can handle explicitly multi-modal, non-normal data.

\begin{table}
    \centering
    \tiny
    \begin{tabular}{c|c|c|c|c|c}
    Sim. & Data Type & Distribution & \makecell{Has\\NAs} & \makecell{Type of\\Regime Change} & Description \\
    \hline
    \hline
    A & Continuous & Unimodal & No & \makecell{Covariance\\Change} & \makecell{Example with a spread \\ change and missing \\ values.} \\
     \hline
     B & Continuous & Unimodal & No & \makecell{Mean\\Change} & \makecell{A shift in the center\\without data issues.}  \\
     \hline
     C & Continuous & Bimodal & No & \makecell{Mean\\Change} & \makecell{Example that moves the \\ modes of the distribution \\ without shifting the \\ overall center.} \\
     \hline
     D & Continuous & Unimodal & Yes & \makecell{Missingness\\Amount Change} & \makecell{Changes in the number\\ of data that are missing.} \\
      \hline
     E & Continuous & Unimodal & Yes & \makecell{Missingness\\Structure Change} & \makecell{Changes in the pattern\\ under which data \\ are missing} \\
     \hline
     F & Mixed & Bimodal & Yes & \makecell{Covariance\\Change} & \makecell{An increase in spread \\ example with multiple \\ data issues.} \\
    \hline
     G & Mixed & Unimodal & Yes & \makecell{Covariance\\Change} & \makecell{Same as previous test \\ case, with a unimodal \\ distribution.} \\
     \hline
     H & Mixed & Unimodal & Yes & \makecell{Mean\\Change} & \makecell{Center shift with \\ data issues.} 
    \end{tabular}
    \caption{Details on the data used in the simulation study described in Figure \ref{fig:simulation_results}.  Many different data examples are chosen to highlight the versatility of this method in detecting a broad class of change-points. }
    \label{tab:sim_types}
\end{table}

The simulation consists of 50 draws (or replications) from each of the data distributions described in Table \ref{tab:sim_types}.  A single data draw includes 30 days of data, where each day has 200 observations.  There are therefore 29 possible change-points, one after each day except the last day.  For each data set, change-points are imposed after day 14.  In each simulation, we record the point estimate of the change-point vector according to the process outlined in Section \ref{subsec:pointestimate}, using a cutoff value of 0.5.

\begin{figure}
    \centering
    \includegraphics[scale = 0.8]{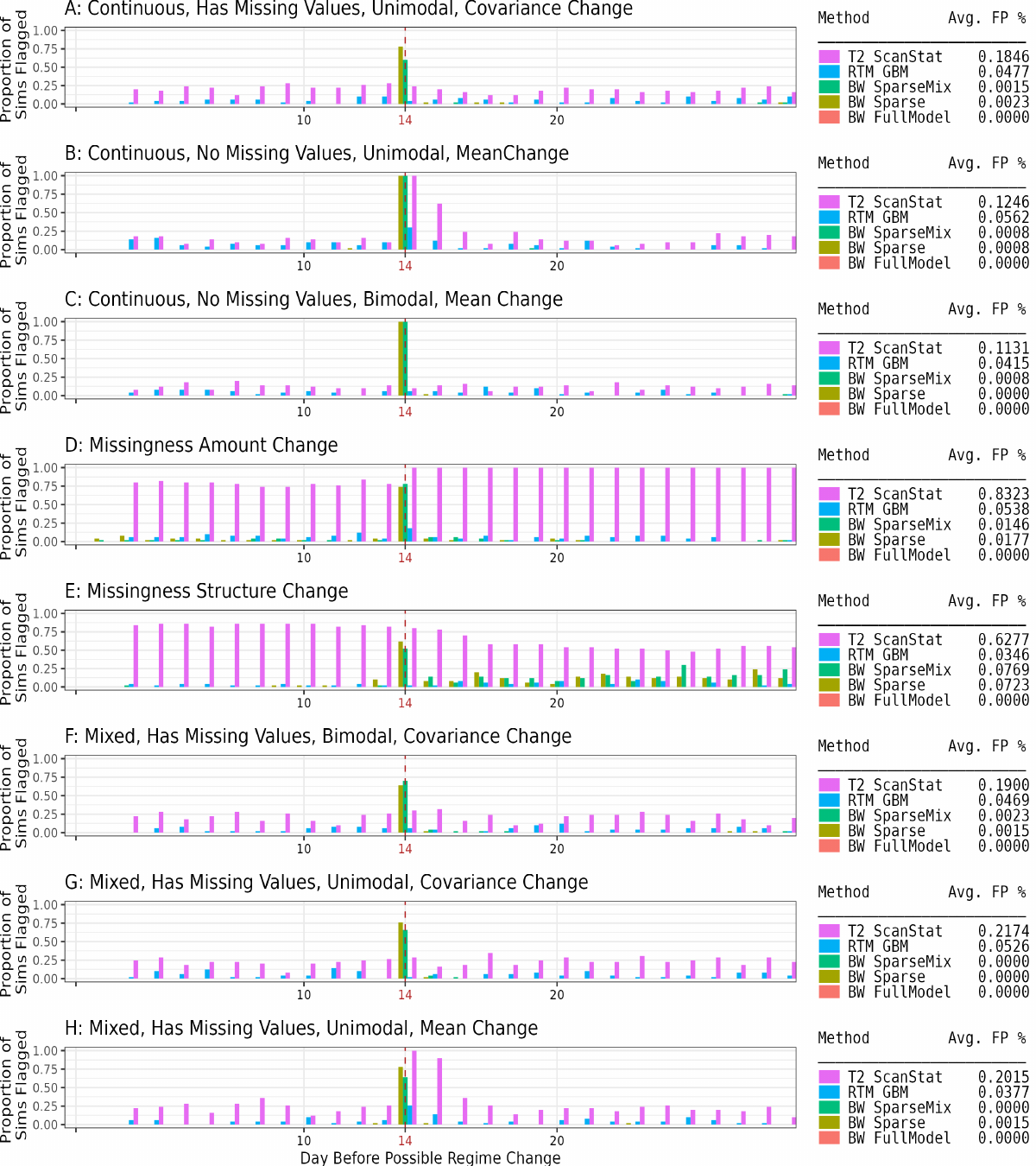}
    \caption{Simulation study over many types of regime changes; each regime change is placed directly after Day 14.  Average False Positive rates are calculated without including the two days after each change-point.}
    \label{fig:simulation_results}
\end{figure}

The results of the simulation study are outlined in Figure \ref{fig:simulation_results}.  In all test cases, the Bayes Watch change-point model performs above the level of the industry standards.  The two versions of the change-point model are the only ones able to detect the bimodal mean change example.  It is perhaps surprising that the single-component model mostly performed as well as its mixture model variant.  More flexibility comes with additional parameters to estimate and increased variance, and is thus not always better.  Since learning the mixture components also comes with a high computational cost, the authors recommend the single-component model for most applications.  Perhaps surprisingly, the single component model without a sparsity assumption (BW FullModel) performed very poorly on all test cases.  We suspect that this is because this model requires significantly more parameters to be accepted at each merge-split step of the MCMC chain, and this led to very low acceptance probabilities.

One drawback noticed in previous iterations of this simulation study is that the presence of major changes in the data sometimes prevent the Bayes Watch change-point model from identifying additional, but more subtle changes that it would normally be able to identify if the larger change were not present.  This drawback is nicely explained by a feature observed in further experiments where the False Positive Rate (FPR) drops whenever the NG-W is centered away from truth.  In these experiments, whenever the center of the NG-W was significantly different than the observed center within a regime, the algorithm was unlikely to accept changes without a signal.  This was especially true whenever there exists a regime that resembles the center of the NG-W much more closely than other regimes.  In this case, the FPR within this regime will match theoretical expectations while the FPR within the other regimes will be lower than expected.  

A major change will skew the initial estimate of the scale matrix in the NG-W prior, making the algorithm less likely to admit subtle and potential change-points.  As a direction of future work, this issue may be solved by replacing the NG-W by a two-component mixture, where one component has a prior distribution heavily concentrated at the empirical center and one component has a much more diffuse prior distribution.  Another direction for future work that may help this issue would be to have the prior distributions for a new regime be centered on the parameter values from the previous regime.

\section{Fault Detection}\label{sec:posterior}
Once a change-point is detected using Bayes Watch, it may not be immediately clear what feature of the data led to this regime change.  A major benefit of the Bayes Watch framework is that it involves fitting highly flexible models to each data regime, which are periodically saved to use for the comparison of the data distributions before and after a regime change.  This is a major boon for the model monitoring task; the method both finds change-points and identifies a change-point's origin.

Assume that a change-point is observed and let $Q_b$ and $Q_a$ be the probability measures for the data before and after the observed change-point, respectively.  Let $\mathcal{H}(Q_b,Q_a)$ be some calculation of distance/divergence between measures $Q_b$ and $Q_a$.  Let $Q_b\backslash \{i\}$ be the marginal distribution of $Q_b$ without variable $i$ and let $Q_b: \{i\}$ be the marginal distribution of variable $i$.  These marginals are particularly easy to calculate with the multivariate normal models fit to the latent data.  Define the following two metrics:
\begin{align*}
    \text{\underline{Total-Effect $\mathcal{H}(Q_b,Q_a)$ Loss of Variable $i$}:}&~ 1 - \frac{\mathcal{H}(Q_b\backslash \{i\},Q_a \backslash \{i\} )}{\mathcal{H}(Q_b,Q_a)}, \\
\text{\underline{First-Order $\mathcal{H}(Q_b,Q_a)$ Loss of Variable $i$}:}&~ \frac{\mathcal{H}(Q_b:\{i\},Q_a:\{i\} )}{\mathcal{H}(Q_b,Q_a)}.
\end{align*} 
The naming of these metrics comes from their near-analogue in variance decomposition procedures used to perform sensitivity analysis \citep{storlie2009}.  Rather than examine the change in model variance that comes from each individual variable, we examine these changes in terms of the function $\mathcal{H}$.  The aim of this approach is to isolate any individual variables that led to a change-point occurring.

There are many methods with which to quantify the difference between two data distributions \citep{gibbs2002}.  For this paper, we experimented with Kullback-Leibler divergence, Jeffrey's divergence, and Hellinger's distance\footnote{The squared Hellinger's Distance is an $f$-divergence using the function $f(t) = (\sqrt{t} - 1)^2$.  Similarly, Jeffrey's is a symmetrization of the KL-Divergence.}.  The authors found that the Hellinger's Distance had the most intuitive and reliable performance of these three.  We proceed with the analysis using Hellinger's Distance as the choice for the function $\mathcal{H}$. 

\begin{figure}
    \centering
    \includegraphics[scale=0.55]{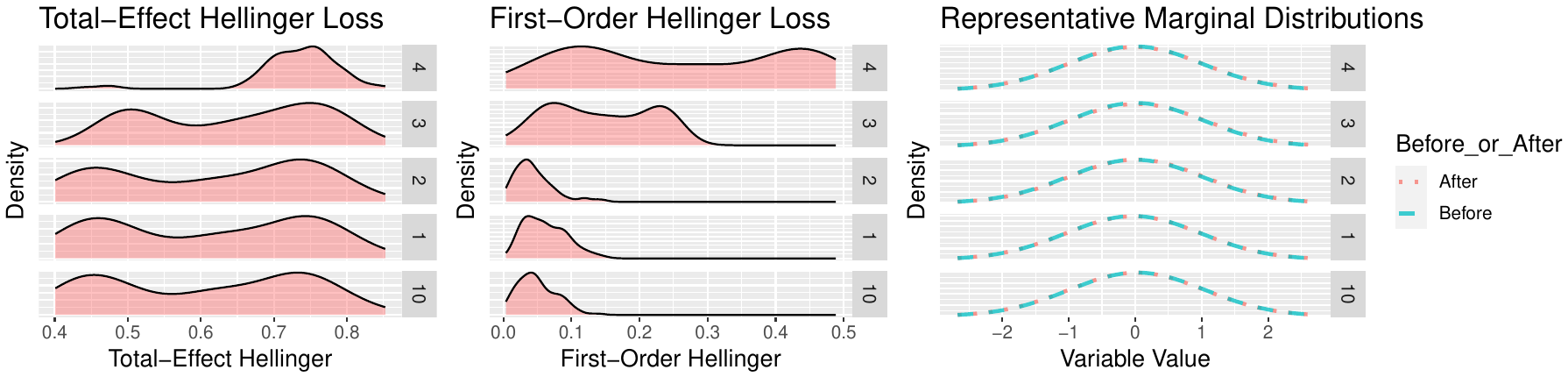}
    \caption{The Total-Effect and First-Order Hellinger distance values recorded every 5 model fits from the MCMC chain applied to Simulation $B$, where variables 3 \& 4 have the only contributions to the change-point, with 4 having the greater contribution.  The far left graph plots the marginal distributions of a represented model fit from the MCMC chain.} 
    \label{fig:posterior_analysis}
\end{figure}

We perform fault detection on a single posterior draw from Simulation $B$, saving the model fits every 5 iterations in the MCMC chain.  To save on space, this model saving scheme only records the models for the last two regimes of any given $\boldsymbol\phi$ sample.  In the case of the mean change simulation, a mean shift was only applied to variables 3 \& 4.  Amongst these two variables, variable 4 saw a greater shift than variable 3.  With these model saves, the Total-Effect and First-Order losses are calculated for every variable, which gives a posterior distribution on these two metrics for every variable.  We examine how these metrics pick out changes in individual variables in Figure \ref{fig:posterior_analysis}, where the 5 posterior distributions with the highest average metric values are graphed.  Both the Total-Effect and First-Order losses are able to identify variable $4$ as having the most significant contribution to the change-point and variable $3$ as the next most significant.  

\begin{figure}
    \centering
   \includegraphics[scale=0.6]{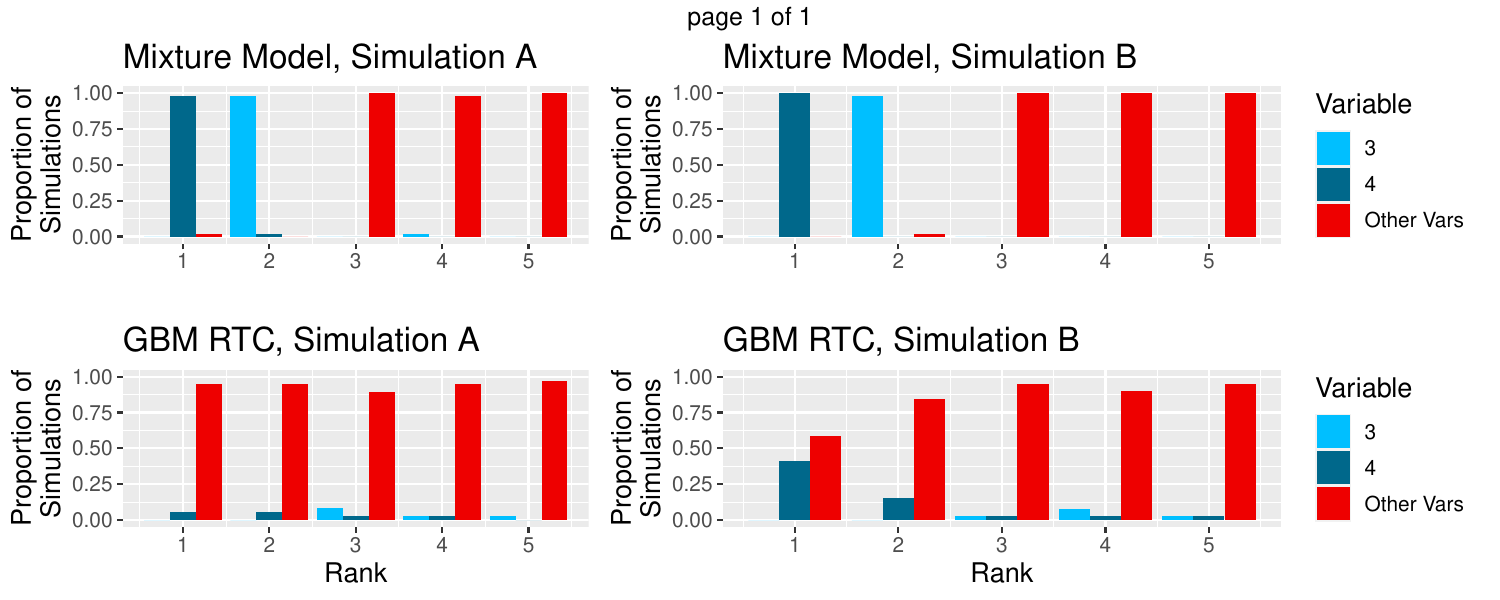}
    \caption{The average rankings of variable importance over a simulation of 50 iterations, assuming that the change-point is known.  The true change happens in variables 3 \& 4, where 4 has the more significant change.  The rankings for all other variables are graphed together in red.  The RTC method with a GBM classifier was largely unable to identify 3 \& 4 as the most important variables.  In contrast, both versions of the Bayes Watch change-point model are able to correctly identify the most and second-most significant variables nearly every time.}
    \label{fig:variable_importance}
\end{figure}

We test the Bayes Watch fault detection system's ability to detect the most significant contributions to a change-point for Simulations $A$ \& $B$.  Note that, similar to Simulation $B$, the covariance change in Simulation $A$ only occurs for variables 3 \& 4, here as an increase in variance, with $4$ seeing the more substantial increase.  At each instance of the simulation, the variables are ranked according to their average Total-Effect Loss across all model saves.  These rankings are compared against the variable importance metrics available for the GBM used in the RTC method.  For the RTC method, each variable is ordered according to its relative influence (calculated using the \textsc{gbm} package in \textsc{R} \citep{greenwell2020}).  These orderings are compared in Figure \ref{fig:variable_importance}.  The Total-Effect Loss for the mixture model and for the single-component model is able to pick out variables 3 \& 4 as the second and first most important variables in nearly every simulation, while the GBM rarely identifies the most significant contributors.

\section{Application to Palliative Care Model}\label{sec:app}
We apply the Bayes Watch framework to two segments of the data fed into the \cite{murphree2021} model, as described in Section \ref{subsec:data}.  For the first segment, we ran the change-point model for 200 iterations of the MCMC chain; a single change-point was found using the default probability cutoff value of 0.5.

\begin{figure}
    \centering
   \includegraphics[scale=0.55]{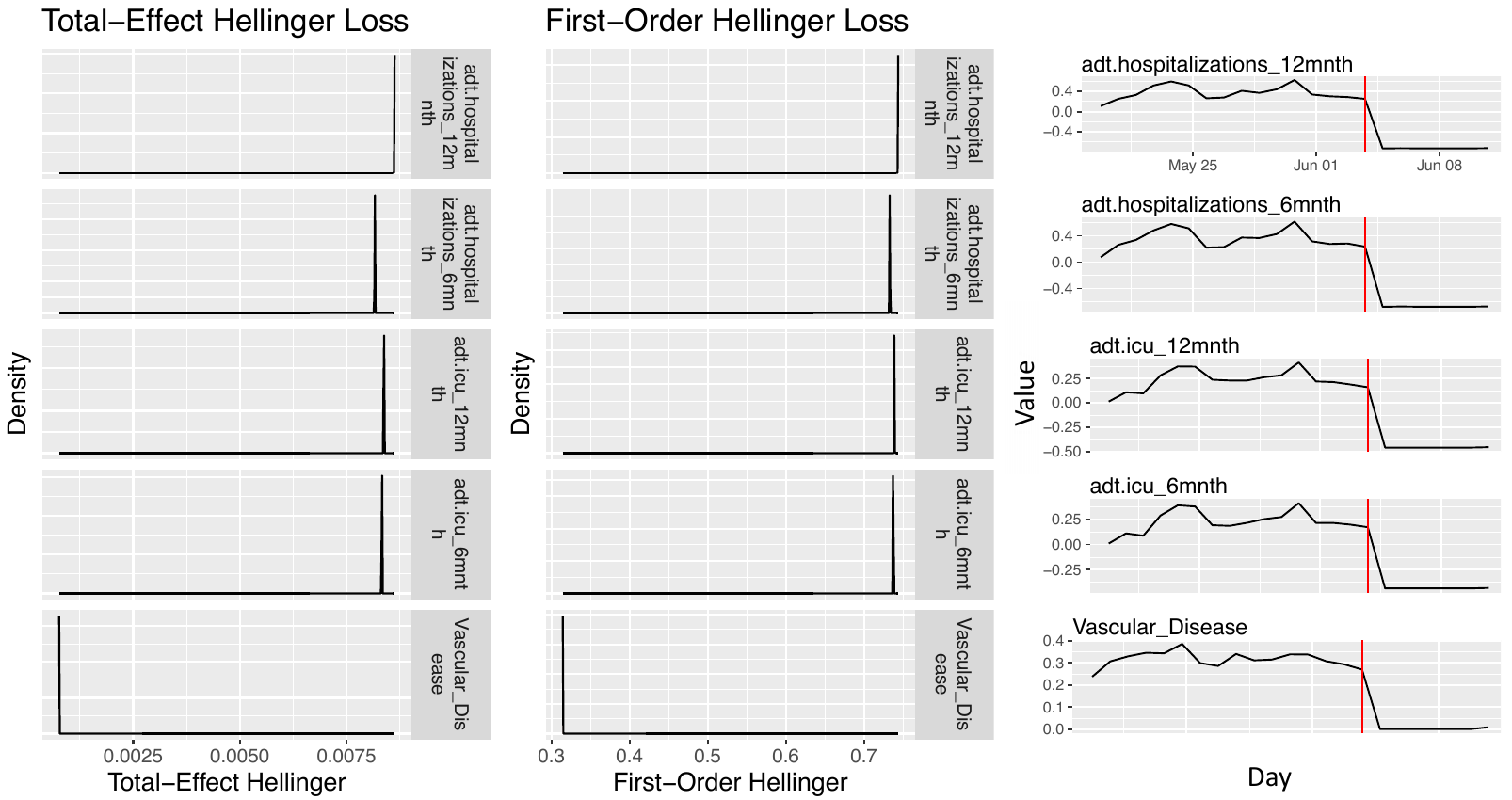}
    \caption{Posterior analysis of change-point found when applying the Bayes Watch framework to the data fed into the \cite{murphree2021} model.  Variables are ranked according to average First-Order Loss.  In lieu of the marginals provided in Figure \ref{fig:posterior_analysis}, we graph the day averages over time on the far right.  The vertical red lines on these timeseries mark the location of the found change-point.  This change-point is likely the result of many variables flattening to zero in early June.  }
    \label{fig:application_posterior}
\end{figure}

As outlined in Figure \ref{fig:application_posterior}, the Bayes Watch change-point model found a single change-point in early June. The variable importance scheme identified that this change-point was likely a result of many different variables flattening to zero. The timeseries on the far right in Figure \ref{fig:application_posterior} show the average daily values of each variable. There is no compelling reason in the data for such a drastic change in these values to occur naturally, which suggests that this change is due to a data pipeline issue, or a change of policy in how data are collected.  This issue was eventually noticed at the time by the palliative care modeling team.  It was a data pipeline issue causing the prior utilization data to no longer be received by the model. It did not cause a model error, as it is permissible for individual patients to have no prior history in our health system.  However, this data pipeline issue was not noticed and not acted upon until several days later.  Had the Bayes watch model been running at the time, it would have detected this issue on day 1 of the change and steps would have been taken to resolve the issue sooner.
 
Notice that ranking the variables according to Total-Effect Loss would give a different ranking than the one seen in Figure \ref{fig:posterior_analysis}, which ranks according to First-Order Loss. We also investigated the Total-Effect ordering, and found that the top eight most significant variables are the same regardless of the metric considered used to order them (the variables ‘Metastatic Cancer and Acute Leukemia,’ ‘Lymphoma and Other Cancers,’ and ‘Seizure Disorders and Convulsions’ are not pictured in Figure \ref{fig:posterior_analysis}). All variables, under both orderings, see a dramatic decrease in average value at the beginning of June 2020. Considering the starkness of the change in value among all of these variables, we would suggest that they all be investigated, regardless of their relative importance according to Total-Effect and First-Order losses.

As discussed in Section \ref{sec:sim}, one drawback of this approach is that the presence of major change-points will prevent Bayes Watch from identifying more subtle changes in the data.  A work-around for this issue would be to subset the data to only include observations after the change-point once a change-point is found, then re-run this algorithm on the subset data.  As mentioned previously, future iterations of this algorithm may use a mixture model prior to allow Bayes Watch to pick up major and minor changes in the same run.

\begin{figure}
    \centering
   \includegraphics[scale=0.55]{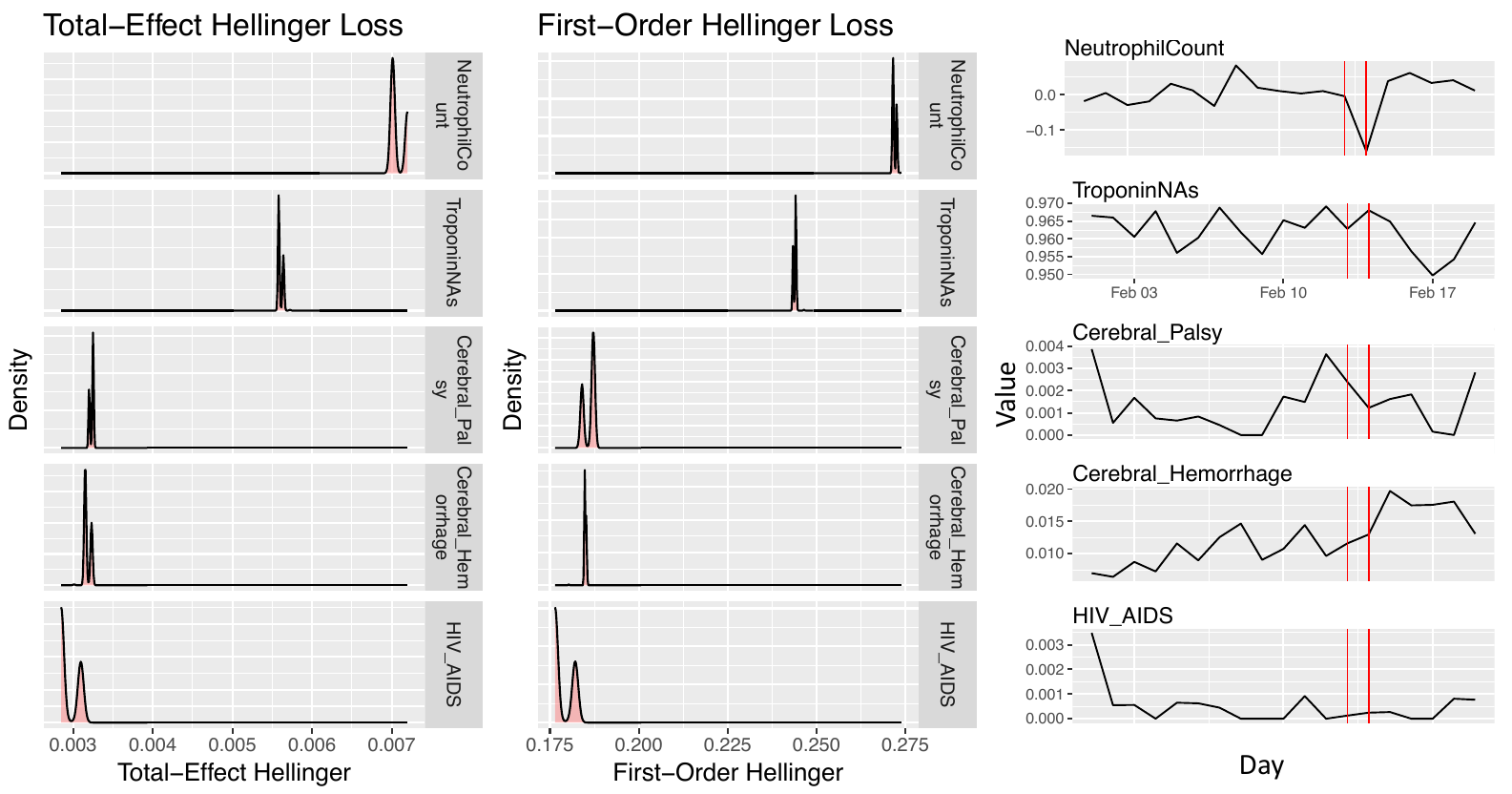}
    \caption{Posterior analysis of change-point found when applying the Bayes Watch framework to the data fed into the \cite{murphree2021} model during February 2023.  Variables are ranked according to average First-Order Loss from the second change-point.  In lieu of the marginals provided in Figure \ref{fig:posterior_analysis}, we graph the day averages over time on the far right.  The vertical red lines on these timeseries mark the location of the found change-points.  This change-point is likely driven by changes in the first two variables: `NeutrophilCount' and `TroponinNAs.'  }
    \label{fig:second_application_posterior}
\end{figure}


We examine a second segment of the palliative care mdoel data, taken from February 2020, prior to many of the shut-downs caused by the COVID-19 pandemic.  Bayes Watch determined two change-points that occurred directly one after another: one on 02/13/20 and one on 02/14/20.  The variables most influential for this particular detection were the missingness of Troponin lab and the value of NeutrophilCount.  This change point was not a known issue for the palliative care model and after further investigation, it is still unclear why Troponin lab values became missing more frequently in late February and stayed that way for an extended period.  To the best of our knowledge there was no change in policy or protocol that would have caused this change.  This would have had an impact on  model scores and performance though, because the presence of a Troponin lab value is an informative feature; the presence of a Troponin lab indicates that the clinician had suspicion of a serious cardiac issue for a patient.  It is possible that this was just a random anomaly as it did eventually resolve, and it is difficult now several years later to uncover the root cause.  However, had the Bayes watch algorithm been running at this time, it would have warranted an investigation.

\section{Discussion}\label{sec:conclusion}
The model-monitoring data challenge presented by the Mayo Clinic data scientists is not necessarily unique to the \cite{murphree2021} model.  Especially in healthcare, there exist many longitudinal datasets with mixed variables types, high dimensions ($p>200$), and missing values.  While the primary aim of this paper is to address the problem of  monitoring predictive models in health care settings, the authors hope that Bayes Watch will become a strong choice for any model monitoring task.  

Bayes Watch is a framework for locating and assessing change-points in a data stream, devised for the purpose of monitoring the performance of in-process models.  This model is extremely versatile, able to identify many different types of changes, both in the center and spread of the underlying data distribution and in the nuanced changes that can occur in the way in which data are missing.  The change-point model is further able to identify regime changes in spite of significant data challenges, such as with missing data and mixed data types.  The simulation study shows the model's superiority over two industry standards: Hotelling's $T^2$ scan statistic and the RTC method using a GBM classifier. 

The analysis of the data fed into the \cite{murphree2021} palliative care model located a major change-point in early June 2020.  Researchers at Mayo had found this change-point, but only well after it had occurred.  This analysis also determined the specific variables that seemed to "flat-line" at the change-point, illustrating how Bayes Watch is able to both locate change-points and diagnosis what changes led to a change-point.  

To investigate more subtle change-points, we also investigated data fed into the \cite{murphree2021} model prior to the major shut-downs that occurred due to the COVID-19 pandemic. We discovered a subtle change-point in late February that would have warranted an investigation by the palliative care model team had Bayes Watch been run on this data stream several years ago when the change occurred.

There are several directions for future work.  As discussed in Section \ref{sec:sim}, the FPR is sensitive to the choice of the hyperpriors in the NG-W distribution.  The choice of hyperpriors affects the change-point model's ability to distinguish subtle changes when stark changes are also present.  One potential solution to this issue is to replace the NG-W distribution by a two-component variant.  In one component, the choice for the scale matrix and degrees of freedom strongly center the distribution at the empirical center, while in the other component the distribution is chosen to be more diffuse.  

As a second direction, the parameter values for the NG-W distribution could be sampled using a random walk centered at the parameter values from the previous regime.  This would be in line with expectations for a time series; the amount that the center and spread in the data distribution changes should be somewhat dependent on where these features were directly before the change.  The present version of the algorithm does not code this dependence, and instead samples new parameter values from a prior that is ignorant of this direct history.

A final direction for future work would be to treat data objects as patients rather than as individual observations on patients, as this would perhaps better represent the actual structure of the data studied in this paper.  While we have shown strong potential for the version of this model that essentially ignores whether individual observations belong to the same patient, we are hopeful that the accuracy and relevance of this model might be improved even more by explicitly modeling the data and trajectory of an entire patient encounter.



\bibliographystyle{rss}
\spacingset{1}
\vspace*{0.5cm}
\bibliography{main}  

\newpage

\begin{center}
    \huge SUPPLEMENTARY MATERIAL
\end{center}

\begin{appendix}
\spacingset{1.75} 

\section{Information on Covariates in Data Application}\label{a:data_vars}

The table in this section is from the Supplementary Material of \cite{murphree2021}.  As mentioned, the analysis in this paper uses the same data.

     \small
    \begin{longtable}{| p{.50\textwidth} | p{.20\textwidth} | p{.20\textwidth} |} 
    \hline
    Variable & Group & Data Type \\
    \hline
    days\_since\_pc & Utilization & continuous \\
    UNIT.TRANSFER & Other & nominal \\
    Troponin & Lab & continuous \\
    Bilirubin & Lab & continuous \\
    Albumin & Lab & continuous \\
    AnionGap & Lab & continuous \\
    NeutrophilCount & Lab & continuous \\
    AST & Lab & continuous \\
    Lactate & Lab & continuous \\
    INR & Lab & continuous \\
    Creatine & Lab & continuous \\
    ALT & Lab & continuous \\
    Phosphate & Lab & continuous \\
    Hct & Lab & continuous \\
    Hb & Lab & continuous \\
    PartPressureArterialO2 & Lab & continuous \\
    Sodium & Lab & continuous \\
    age & Demographics & continuous \\
    Ammonia & Lab & continuous \\
    Metastatic\_Cancer\_and\_Acute\_Leukemia & Comorbidity & binary \\
    adt.hostitalizations\_12mnth & Utilization & continuous \\
    adt.hospitalizations\_6mnth & Utilization & continuous \\
    ICUTRANSFER & Other & binary \\
    Potassium & Lab & continuous \\
    Calcium & Lab & continuous \\
    Platelets & Lab & continuous \\
    BicaronateArterialBloodGas & Lab & continuous \\
    BaseDeficit & Lab & continuous \\
    pHArterialBloodGas & Lab & continuous \\\
    adt.icu\_12mnth & Utilization & continuous \\
    Lung\_and\_Other\_Severe\_Cancers & Comorbidity & binary \\
    adt.icu\_6mnth & Utilization & continuous \\
    APTT & Lab & continuous \\
    cameFrom & Demographics & nominal \\
    Lipase & Lab & continuous \\
    community & Demographics & binary \\
    Pressure\_Pre\_Ulcer\_Skin\_Changes\_or\_Un & Comorbidity & binary \\
    Septicemia\_Sepsis\_Systemic\_Inflammato & Comorbidity & binary \\
    Congestive\_Heart\_Failure & Comorbidity & binary \\
    Aspiration\_and\_Specified\_Bacterial\_Pn & Comorbidity & binary \\
    Specified\_Heart\_Arrhythmias & Comorbidity & binary \\
    Disorders\_of\_Immunity & Comorbidity & binary \\
    Acute\_Renal\_Failure & Comorbidity & binary \\
    Cardio\_Respiratory\_Failure\_and\_Shock & Comorbidity & binary \\
    Dementia\_Without\_Complication & Comorbidity & binary \\
    Chronic\_Kidney\_Disease\_Severe\_State\_4 & Comorbidity & binary \\ 
    Pressure\_Ulcer\_of\_Skin\_with\_Full\_Thic & Comorbidity & binary \\
    Chronic\_Ulcer\_of\_Skin\_Except\_Pressure & Comorbidity & binary \\ 
    Hip\_Fracture\_Dislocation & Comorbidity & binary \\
    Chronic\_Obstructive\_Pulmonary\_Disease & Comorbidity & binary \\
    Acute\_Myocardial\_Infraction & Comorbidity & binary \\
    Amylase & Lab & binary \\
    Diabetes\_with\_Chronic\_Complications & Comorbidity & binary \\
    Chronic\_Kidney\_Disease\_Moderate\_Stage & Comorbidity & binary \\
    Lymphoma\_and\_Other\_Cancers & Comorbidity & binary \\
    Opportunisitic\_Infections & Comorbidity & binary \\
    Atherosclerosis\_of\_the\_Extremities\_wi & Comorbidity & binary \\
    Pneumococcal\_Pneumonia\_Empyema\_Lunch\_A & Comorbidity & binary \\
    Coagulation\_Defects\_and\_Other\_Specifi & Comorbidity & binary \\
    Chronic\_Kidney\_Disease\_Mild\_or\_Unspec & Comorbidity & binary \\
    Vertebral\_Fractures\_without\_Spinal\_co & Comorbidity & binary \\
    Severe\_Hematological\_Disorders & Comorbidity & binary \\
    Protein\_Calorie\_Malnutrition & Comorbidity & binary \\
    Dialysis\_Status & Comorbidity & binary \\
    Cirrhosis\_of\_Liver & Comorbidity & binary \\
    Pressure\_Ulcer\_of\_Skin\_with\_Partial\_T & Comorbidity & binary \\
    Intestinal\_Obstruction\_Perforation & Comorbidity & binary \\
    Polyneuropathy & Comorbidity & binary \\
    Nephritis & Comorbidity & binary \\
    Fibrinogen & Lab & continuous \\
    Vascular\_Disease & Comorbidity & binary \\
    Amputation\_Status\_Lover\_Limb\_Amputati & Comorbidity & binary \\
    Breast\_Prostate\_and\_Other\_Cancers\_and & Comorbidity & binary \\
    Exudative\_Mascular\_Degeneration & Comorbidity & binary \\
    Pressure\_Ulcer\_of\_Skin\_with\_Necrosis\_ & Comorbidity & binary \\
    Fibrosis\_of\_Lung\_and\_Other\_Chronic\_Lu & Comorbidity & binary \\
    Seizure\_Disorders\_and\_Convulsions & Comorbidity & binary \\
    Unstable\_Angina\_and\_Other\_Acute\_Ische & Comorbidity & binary \\
    Major\_Organ\_Transplant\_or\_Replacement & Comorbidity & binary  \\
    Complications\_of\_Specified\_Implanted\_ & Comorbidity  & binary \\
    Proliferative\_Diabetic\_Retinopathy\_an & Comorbidity & binary \\
    Bone\_Joint\_Muscle\_Infections\_Necrosis & Comorbidity & binary \\
    Colorectal\_Bladder\_and\_Other\_Cancers & Comorbidity & binary \\
    Morbid\_Obesity & Comorbidity & binary \\
    Unspecified\_Renal\_Failure & Comorbidity & binary \\
    Inflammatory\_Bowel\_Disease & Comorbidity & binary \\
    Traumatic\_Amputations\_and\_Complications & Comorbidity & binary \\
    Major\_Depressive\_Bipolar\_and\_Paranoid & Comorbidity & binary \\
    Artificial\_Openings\_for\_Feeding\_or\_El & Comorbidity & binary \\
    Other\_Significant\_Endocrine\_and\_Metab & Comorbidity & binary \\
    Schizophrenia & Comorbidity & binary \\
    Parkinsons\_and\_Huntingtons\_Diseases & Comorbidity & binary \\
    Chronic\_Hepatitis & Comorbidity & binary \\
    gender & Demographics & binary \\
    Respiratory\_Arrest & Comorbidity & binary \\
    DAYCOV & Other & numeric \\
    Rheumatoid\_Arthritis\_and\_Inflammatory & Comorbidity & binary \\
    Ischemic\_or\_Unspecified\_Stroke & Comorbidity & binary \\
    Multiple\_Sclerosis & Comorbidity & binary \\
    Chronic\_Pancreatitis & Comorbidity & binary \\
    Respirator\_Dependence\_Tracheostomy\_St & Comorbidity & binary \\
    Dementia\_With\_Complications & Comorbidity & binary \\
    Cerebral\_Palsy & Comorbidity & binary \\
    Hemiplegia\_Hemiparesis & Comorbidity & binary \\
    Cerebral\_Hemorrhage & Comorbidity & binary \\
    Diabetes\_without\_Complication & Comorbidity & binary \\
    HIV\_AIDS & Comorbidity & binary \\
    Coma\_Brain\_Compression\_Anoxic\_Damage & Comorbidity & binary \\
    Quadriplegia & Comorbidity & binary \\
    Muscular\_Dystrophy & Comorbidity & binary \\
    Paraplegia & Comorbidity & binary \\
    Drug\_Alcohol\_Dependence & Comorbidity & binary \\
    Amyotrophic\_Lateral\_Sclerosis\_and\_Oth & Comorbidity & binary \\
    Monoplegia\_Other\_Paralytic\_Syndromes & Comorbidity & binary \\
    Spinal\_Cord\_Disorders\_Injuries & Comorbidity & binary \\
    Severe\_Skin\_Burn\_or\_Condition & Comorbidity & binary \\
    Drug\_Alcohol\_Psychosis & Comorbidity & binary \\
    Cystic\_Fibrosis & Comorbidity & binary \\
    Angina\_Pectoris & Comorbidity & binary \\
    Major\_Head\_Injury & Comorbidity & binary \\
    Diabetes\_with\_Acute\_Complications & Comorbidity & binary \\
    \hline
    \caption{Information on the variables described in Section \ref{subsec:data}.  This table is taken directly from \cite{murphree2021}. }
    \end{longtable}
    \normalsize
%

\section{MCMC Chain Details}\label{a:mcmc}
In this section, specific details on the Gibbs sampler in Section \ref{sec:mcmc} are provided.

\subsection{Drawing \texorpdfstring{$\mathbf{Z}$}{}}  Given an estimate for the regime vector $\boldsymbol\phi$ and the parameters $\boldsymbol\Theta$, one can sample the latent data for missing values, for discrete variables, and for observations that take values equal to a variable's censored bounds.  Consider an observation $\mathbf{y}_{t;i}$, and let $\mathcal{M}_{t;i}, \mathcal{B}_{t;i} \subseteq \{ 1, \dots, J\}$ be disjoint sets of indices such that
\[ y_{t;i,j} ~~\begin{cases}
\text{is missing} & \text{if $j \in \mathcal{M}_{t;i}$} \\
\text{is equal to $b_j$ or $c_j$} & \text{if $j \in \mathcal{B}_{t;i}$}
\end{cases} \]
First consider all $\mathbf{z}_{t;i, \mathcal{M}_{t;i}}$.  Since the parameters, regime assignments, and component assignments are assumed to be known, draw $\mathbf{z}_{t;i, \mathcal{M}_{t;i}}$ from
\begin{align*}
    &\hspace{-0.5cm}\mathbf{z}_{t;i, \mathcal{M}_{t;i}} \Big| \mathbf{z}_{t;i} \backslash \mathbf{z}_{t;i, \mathcal{M}_{t;i}}, \boldsymbol\theta^{(r)}, \boldsymbol\phi, \boldsymbol\gamma \sim \mathcal{N} \biggl((\mu^{(r)}_{\gamma_{t;i} })_{\mathcal{M}_{t;i} } + \\
    & (\Lambda^{(r)}_{\gamma_{t;i} })_{\mathcal{M}_{t;i},\mathcal{M}_{t;i} }^{-1}(\Lambda^{(r)}_{\gamma_{t;i}})_{\mathcal{M}_{t;i}, [J] \backslash \mathcal{M}_{t;i} }( \mathbf{z}_{t;i} \backslash \mathbf{z}_{t;i, \mathcal{M}_{t;i}} - (\mu^{(r)}_{\gamma_{t;i}})_{[J] \backslash \mathcal{M}_{t;i}} ), ~~ (\Lambda^{(r)}_{\gamma_{t;i}})_{\mathcal{M}_{t;i},\mathcal{M}_{t;i}}^{-1} \biggr).
\end{align*}

After drawing $\mathbf{z}_{t;i, \mathcal{M}_{t;i}}$, draw the latent variables associated with the binary data values.  For each individual element $\mathbf{z}_{t;i, j'}, j' \in \mathcal{D}$, draw from the univariate conditional distribution 
\begin{align}\label{eq:univariateconditional}
   &\hspace{-0.2cm}z_{t;i, j'} | \mathbf{z}_{t;i,[J] } \backslash z_{t;i, j'}, \boldsymbol\theta^{(r)}, \boldsymbol\phi, \boldsymbol\gamma \sim \\
   &\hspace{0.2cm} \mathcal{N} ((\mu^{(r)}_{\gamma_{t;i} })_{j'} + (\Lambda^{(r)}_{\gamma_{t;i}})_{j',j'}^{-1}(\Lambda^{(r)}_{\gamma_{t;i} })_{j', [J] \backslash \{j'\} }( \mathbf{z}_{t;i} \backslash  z_{t;i, j'} - (\mu^{(r)}_{\gamma_{t;i} })_{[J] \backslash \{j'\} } ), (\Lambda^{(r)}_{\gamma_{t;i}})_{j',j'}^{-1} ), \nonumber
\end{align} 
truncated to the set $\{ z_{t;i, j'} \geq 0 \}$ if $y_{t;i,j} = 1$, and $\{ z_{t;i, j'} < 0 \}$ otherwise.  

Lastly, consider the latent data associated with observations that achieve a boundary: $z_{t;i,j}$ where $j \in \mathcal{B}_{t;i}$.  For every such element $z_{t;i,j}$, also draw from the distribution in Equation \ref{eq:univariateconditional}, but instead truncated to be above or below the right or left boundary, depending on whether $y_{t;i,j} = b_j$ or $y_{t;i,j} = c_j$.

\subsection{Drawing \texorpdfstring{$\boldsymbol\Theta, P, w, v$}{}}\label{sec:drawTheta}

The parameter sets $\boldsymbol\theta^{(r)}$ are drawn using the Mixture Model described in Section \ref{subsec:dirichlet}.  When the component assignments $\boldsymbol\gamma$ are known, this prior is conjugate assuming MVN data and fixed graph structure $G$.  Although this setup allows for an unbounded number of regimes, to simulate this in practice choose the number of regimes $R$ to be well above the number of regimes one would expect to observe.

Calculating proposal matrices for the precision matrices in the parameter set $\boldsymbol\theta^{(r)}$, given a fixed graph structure $G$, is a non-trivial task.  Each $\Lambda$ matrix proposed must be from a random process that produces a matrix that matches the zero structure encoded in $G$, while also remaining positive definite.  This task is performed using the algorithm proposed by \cite{lenkoski2013}, where a sample $\Lambda$ is first taken from the full Wishart distribution, then iteratively updated until $\Lambda \in \mathcal{P}_G$ and $\Lambda \sim \mathcal{W}_G(D, \nu)$.  A key note about this process is that the resulting matrix is simulated from the $G$-Wishart distribution itself, rather than some other proposal distribution that produces a valid matrix.  This is necessary for the update to remain conjugate.

Redrawing the elements of the probability transition matrix $P$ is a conjugate update, since they are each assumed to follow Beta$(w,v)$ distributions.  Redrawing the hyperparameters $w, v$ is also a conjugate update, since they are assumed to follow Gamma$(a_c,b_c)$, $c \in \{ w, v\},$ distributions.

\subsection{Drawing \texorpdfstring{$\boldsymbol\phi$}{}}\label{sec:drawPhi}
To calculate good proposals for $\boldsymbol\phi$, we take a split-merge-swap approach.  When a merge is undertaken, two adjacent regimes are combined into a single regime; when a split is undertaken, a regime is split into two adjacent regimes.  In the third step of the MCMC chain (in which the regime vector is resampled), a merge or split action is chosen randomly, then regimes are relabeled across the model to ensure that they remain sequentially ordered.  After a merge or split is attempted, this algorithm iterates through every change-point and resamples the regime assignments for the days at the beginning and end of every regime (constrained in such a way that this swap action can neither add or remove a regime).

The above algorithm has a structure similar to the split-merge-shuffle algorithm found in \cite{martinez2014}.
However, unlike \cite{martinez2014}, a procedure is added to estimate the discrete probability distribution on possible points at which to split a regime into two separate regimes.  This procedure gives weights to split points that are proportional to the likelihood gain of fitting the model on the split regime vector.  This is in contrast to choosing a split point uniformly at random, which may slow the convergence of the MCMC chain in cases where the probability of a regime change is concentrated to a small group of sequential change-points.  To mitigate the computational burden of this split-finding algorithm, a stratified sampling of possible cutpoints is used when calculating the proposal distribution of a split, or when determining the return probability for a merge.  This sampling process approximates the likelihood value of every single possible cutpoint using a grid, which allows the sampler to find timepoints with a high probability of split quickly.   A full pseudocode of the algorithm is available in Section \ref{a:alg}.

Let $\Phi$ be the space of possible regime vectors given the maximum number of possible regimes $R$.  Whenever an update for $\boldsymbol\phi$ is proposed, new values are needed for the parameter collections $\boldsymbol\Theta$, since an update to $\boldsymbol\phi$ causes an update to each of these parameters' posterior distributions.  To improve mixing, instead of drawing from the distribution for $\boldsymbol\phi | \mathbf{Z}, \boldsymbol\Theta, \boldsymbol\gamma, P,  \mathbf{Y}$, one should draw from $\boldsymbol\phi, \boldsymbol\Theta | \mathbf{Z}, \boldsymbol\gamma, P, \mathbf{Y}$.  Although a conjugate prior has been chosen for the collection $\boldsymbol\Theta$, when $\boldsymbol\phi$ changes the underlying \textit{data assignment} for these parameters also changes, which results in lingering terms in the final MH ratio.  To motivate this, consider a proposal to merge regimes $r$ and $r+1$ into regime $r$, denoted by moving from a regime vector $\boldsymbol\phi$ to $\tilde{\boldsymbol\phi}$.  The posterior $G$-Wishart distribution (given the regime vectors) is used as the proposal distribution for the model parameters.  Then, the Metropolis-Hastings ratio (assuming for the time being that the data all belong to the same mixture component) would be
\begin{align}\label{eq:wrongMH}
      \frac{ \left(\prod_{r=1}^M p(\tilde{\mathbf{z}}^{(r)} | \tilde{\theta}^{(r)})\right) \prod_{r=1}^R p(\tilde{\theta}^{(r)} | \tilde{\boldsymbol\phi} ) \left(\prod_{r=1}^{M+1} p( \theta^{(r)} |\mathbf{z}^{(r)} )\right) \prod_{r=M+1}^R p(\theta^{(r)}) p(\tilde{\boldsymbol\phi} | P) \mathfrak{r}(\boldsymbol\phi) }{\left( \prod_{r=1}^{M+1}  p(\mathbf{z}^{(r)} | \theta^{(r)}) \right) \prod_{r=1}^R p(\theta^{(r)} | \boldsymbol\phi ) \left(\prod_{r=1}^M p( \tilde{\theta}^{(r)} |\tilde{\mathbf{z}}^{(r)} )\right) \prod_{r=M}^R p(\tilde{\theta}^{(r)}) p(\boldsymbol\phi | P)\mathfrak{r}(\tilde{\boldsymbol\phi}) } 
\end{align} 
where $M$ is the number of observed regimes before updating, $\mathbf{z}^{(r)}$ \& $ \tilde{\mathbf{z}}^{(r)}$ refer to the data associated with regime $r$ according to $\boldsymbol\phi$ and $\tilde{\boldsymbol\phi},$ respectively, and $\mathfrak{r}(\cdot)$ is the sampling distribution on $\boldsymbol\phi$.  Equation \eqref{eq:wrongMH} reduces to
\begin{align}\label{eq:reducedMH}
      \frac{I(D + \mathbf{z}^{(r)} (\mathbf{z}^{(r)})', \nu + n^{(r)}) I(D + \mathbf{z}^{(r+1)} (\mathbf{z}^{(r+1)})', \nu + n^{(r+1)})}{I(D +\tilde{\mathbf{z}}^{(r)} (\tilde{\mathbf{z}}^{(r)})', \nu + \tilde{n}^{(r)})I(D, \nu)} \cdot \frac{p(\tilde{\boldsymbol\phi} | P) \mathfrak{r}(\boldsymbol\phi) }{p(\boldsymbol\phi | P)\mathfrak{r}(\tilde{\boldsymbol\phi})  },
\end{align}
where $n^{(r)}$ is the number of observations in regime $r$ according to state vector $\boldsymbol\phi$ and $\tilde{n}^{(r)}$ is the number of observations in regime $r$ according to the next state vector $\tilde{\boldsymbol\phi}$.  A similar ratio occurs when performing a split action instead of a merge.  This ratio in Equation \eqref{eq:reducedMH} represents a potential drawback of this approach.  Whenever the underlying data of regime $r$ changes via a merge or split on that regime, it leaves a ratio of $G$-Wishart normalizing constants that would cancel (due to conjugacy) if the data assignments were unchanged.

A merge or split on $\boldsymbol\phi$ requires a reassignment of the components in $\boldsymbol\gamma$ associated with the data values that change their regime.  For simplicity, the component assignments from a previous regime vector are carried over to the new regime vector.  Under the Gibbs framework, reassignments can be made by resampling the component assignments shortly after a given merge/split is accepted or rejected.   

\subsection{Drawing \texorpdfstring{$\boldsymbol\gamma$}{}}
Due to the truncated stick breaking process described in Section \ref{subsec:dirichlet}, the vector $\boldsymbol\gamma$ can updated by a straightforward Gibbs sampler.  However, this approach has known mixing issues.  Thus, we update $\boldsymbol\gamma$ according to a modified split-merge algorithm similar to the ones proposed in \cite{storlie2018} and \cite{kim2006}.  Full pseudocode of this algorithm is available in Section \ref{a:alg}.

\subsection{Drawing Hyperparameters \texorpdfstring{$m$, $\lambda$}{} }\label{subsec:hyperparams}

Due to an absence of prior knowledge on the hyperparameters $\lambda$ and $m$, we embrace a Bayesian hierarchical framework to learn these values as well.  To leverage the computational simplicity of conjugate updates, the following normal and gamma priors are chosen:
\[ m \sim \mathcal{N}_J (0, \mathcal{I}_J) ~~;~~ \lambda \sim \Gamma(c,d),\]
where $\mathcal{I}_J$ is a $J$-dimensional identity matrix.
Given the observed regime model parameters $\{ \boldsymbol\theta^{(1)}, \dots, \boldsymbol\theta^{(R)}\}$, this leads to the posteriors
\begin{align*}
    &m \Big| \{ \boldsymbol\theta^{(1)}, \dots, \boldsymbol\theta^{(R)} \}, \lambda \sim \mathcal{N}_J  \biggl( \left( \mathcal{I}_J + \sum_{r=1}^R \sum_{i=1}^Q \lambda * \Lambda^{(r)}_i \right)^{-1} \left( \sum_{r=1}^R \sum_{i=1}^Q \lambda \Lambda^{(r)}_i \mu^{(r)}_i \right), \\
    &\hspace{5cm}\left( \mathcal{I}_J + \sum_{r=1}^R \sum_{i=1}^Q \lambda * \Lambda^{(r)}_i \right)^{-1} \biggr),\\
    &\lambda \Big| \{ \boldsymbol\theta^{(1)}, \dots, \boldsymbol\theta^{(R)} \}, m \sim \Gamma \left( c + \frac{RJQ}{2}, \frac{1}{2} \sum_{r=1}^R \sum_{i=1}^Q \tr \left( (\mu^{(r)}_i - m)(\mu^{(r)}_i - m)' \Lambda^{(r)}_i \right) + d \right).
\end{align*}

\subsection{Drawing Hyperparameter \texorpdfstring{$D$}{} }\label{subsec:hyperparams2}

One potential solution to this issue is to put a Wishart prior on the scale matrix $D$ from \eqref{eq:NGW}, $D \sim \mathcal{W}\left(\eta, F \right)$ for degrees of freedom $\eta$ and scale matrix $F$.  Since the distributions on the $\Lambda^{(r)}$'s are $G$-Wishart distributions, this update is generally not conjugate.  We choose a sampling distribution on $D$ so that this update will be conjugate whenever $G$ is a full graph,
\[ r\left( D | \boldsymbol\Theta, \boldsymbol\phi, F, \eta, \nu\right)  \sim \mathcal{W}\left(RQ(\nu + J - 1) + \eta, F + \sum_{r=1}^R \sum_{q=1}^Q \Lambda_q^{(r)} \right). \]
The authors tested this approach and concluded that it is computationally expensive and does not significantly impact the FPR issue.  Thus, the final version of the algorithm fixes the scale matrix of the NG-W at the observed empirical center (given the initial degrees of freedom).  

As a direction of future work, this issue may be solved by replacing the NG-W by a two-component mixture, where one component has a prior distribution heavily concentrated at the empirical center and one component has a much more diffuse prior distribution. 

\subsection{Drawing \texorpdfstring{$G$}{}}
To propose a new graph structure $\tilde{G}$ given a current graph structure $G$, only graphs within a neighborhood of $G$ are considered.  That is, a graph proposal is simulated according to the following procedure: select an edge from the space of possible edges, then either remove that edge if it already exists in $G$, or add that edge if it is currently absent.  While this proposal sampling technique is relatively simple, accepting or rejecting this proposal is one of the major challenges of the Bayes Watch framework.  Since the graph structure $G$ determines the zero structure of the precision matrix $\Lambda$, updating $G$ by adding or subtracting an edge equates to a change in the number of free parameters in the current $\Lambda$ matrices, which demands a transdimensional joint update of the $G$ and $\Lambda$ parameters, where the \textit{number} of parameters changes in addition to the values themselves.

Classically, the transdimensional problem of adding and removing edges in a GGM is solved using a Reversible Jump (RJ) MCMC framework.  The RJ uses a Jacobian for the change of parameters that happens when one adds or drops dimensions \citep{green1995}.  The calculation of this Jacobian can often be nontrivial, and requires careful thought from the researcher as to how to structure their MCMC chain.  In \cite{roverato2002}, a method is proposed that uses the RJ framework to search the joint parameter space of precision matrices and graph structures, $(\Lambda, G)$.  This method achieves a trivial RJ Jacobian by proposing samples on the Cholesky decomposition on $\Lambda$.  This approach is explored further in \cite{dobra2011}, who find that when using a $G$-Wishart prior, a ratio of $G$-Wishart normalizing constants appear in the final MH ratio.  For an outline of \cite{dobra2011}'s algorithm, including an explanation as to why this ratio of normalizing constants appears, see Section \ref{a:rjmcmc}.

The Bayes Watch framework uses the Double Reversible Jump (DRJ) algorithm of \cite{lenkoski2013} to efficiently resample the graph structure.  The DRJ combines an exact sampler of the $G$-Wishart distribution, the RJ procedure in \cite{dobra2011}, and the exchange algorithm of \cite{wang2012}, to create a MCMC procedure that sidesteps the need to calculate the $G$-Wishart normalizing constants, since this computation is the major computational bottleneck.  The algorithm in \cite{wang2012} involves creating an extended parameter space and performing an ``exchange" step \citep{murray2006}; this results in a joint chain in which clever draws from the prior distribution cause the $G$-Wishart normalizing constants to drop out, yet leave the marginals of the original joint parameter distribution, $(\Lambda, G)$, unchanged.  We provide a slight update to the DRJ algorithm to also simultaneously draw from the mean parameter $\mu$ (see Section \ref{a:rjmcmc}).  

\subsection{Comments on the Computational Challenges of the G-Wishart}  

The $G$-Wishart normalizing constant is famously difficult to compute in practice; methods for computing this constant, or avoiding its calculation entirely, are popular topics in modern research.  While a direct calculation of the $G$-Wishart normalizing constant, $I_G(\nu, D)$, does exist, it is computationally inefficient, and researchers tend to favor approximation when it absolutely must be calculated \citep{uhler2018}.  Common approaches to this approximation include additional MCMC sampling \citep{atay-kayis2005} and Laplace approximations \citep{lenkoski2011}.  The MCMC algorithm has the drawback of taking longer to converge for larger values of the degrees of freedom parameter $\nu$, while the Laplace approximation becomes less accurate for smaller values of $\nu$.  Fortunately, these two drawbacks complement one another, so the MCMC technique is used to estimate the prior normalizing constant, while the Laplace approximation is used to estimate the posterior.  This aligns with methods suggested in \cite{lenkoski2011} and \cite{mohammadi2021}.  For further reading on this problem, and the creative solutions available in the literature, the authors would additionally suggest the papers by \cite{mohammadi2015, mohammadi2021, wang2012, lenkoski2013}. 

Through many experiments, the authors found that the code for the full graph model (which does not require the calculation of a $G$-Wishart normalizing constant) occasionally takes around the same amount of time as the model for a general sparse graph.  The slow speeds were due to the slow convergence of the Iterative Proportion Scaling algorithm developed by \cite{speed1986} that is used in the Laplace approximation from \cite{lenkoski2011}.  For many data applications with moderate data sizes ($p < 100, n < 1\text{e}6$), the authors would still recommend the sparse model with the Laplace approximation, since (as discussed in Section \ref{sec:sim}) mixing issues with the full model make the sparse model much preferable.  When speed is a major concern, it is possible to avoid the need to calculate the Laplace approximation with the sparse model by only considering decomposable graphs \citep{letac2007}.  As proven in \cite{roverato2002}, the normalizing constant of a $G$-Wishart distribution can be factorized according to $G$'s prime components and their separators.  When these prime components are all complete graphs, which is the case for decomposable graphs, the $G$-Wishart normalizing term can be written as a ratio of Wishart normalizing terms.  The authors found that the decomposable restriction on the graph structure yielded significantly faster runtimes.

\section{Double Reversible Jump Metropolis-Hastings}\label{a:rjmcmc}

The aim of this supplement is to provide an explanation of the Double Reversible Jump algorithm for researchers interested in better understanding it.  Recall from the main paper that the primary aim of this approach is to resample the graph structure without having to calculate the untenable $G$-Wishart normalizing term in the Metropolis-Hastings ratio.  At the end of this supplement, we also discuss our slight update to this algorithm that extends the process for the NG-W.  

We will closely follow the setup of \cite{lenkoski2013}, which is combination of the reversible jump algorithm proposed by \cite{dobra2011} and the exchange algorithm proposed by \cite{murray2006}.  We develop a slight modification to this algorithm that allows for draws from the joint parameter space $(\mu, \Lambda, G)$, rather than just $(\Lambda, G)$.  This modification has the added benefit of simplifying the Metropolis-Hastings ratio.  

Start by considering the upper triangular Cholesky decomposition matrix $\boldsymbol\Psi$ such that $\boldsymbol\Psi'\boldsymbol\Psi = \Lambda$.  Working in the space of the elements $\boldsymbol\Psi$ has the advantage that the free elements of $\boldsymbol\Psi$ can be used to ``complete" the rest of the matrix, given a fixed graph structure $G$, by way of a simple formula.  Fixing $G$, if we knew every element $\Psi_{ij}$ for $(i,j) \in G$, then all elements $\Psi_{ij}$ for $(i,j) \notin G$ can be derived.  This is done by taking every $(i,j) \notin G, 2 \leq i < j$, and setting \begin{equation}\label{eq:cholesky}
    \Psi_{ij} := - \frac{1}{\Psi_{ii}} \sum_{l=1}^i \Psi_{li}\Psi_{lj},
\end{equation}
and $\Psi_{1, j} := 0$ for every $j \geq 2$ when $(1,j) \notin G.$  This completion process is due to \cite{roverato2002}, who makes the point that the process can be parallelized for each row.  An efficient implementation of this process using the C++ language is available within the `bayesWatch' package.  Thus, we can draw the free elements of the matrix $\boldsymbol\Psi$ that correspond to $(i,j) \in G$ randomly, then complete $\boldsymbol\Psi$ efficiently using Equation \ref{eq:cholesky} so that the resulting $\Lambda$ matrix has the correct zero structure and is guaranteed to be positive definite \citep{roverato2002}. 

The algorithm by \cite{dobra2011} begins by proposing a new graph by either removing or adding a single edge.  Without loss of generality, suppose the new graph is $G^e$, which is assumed to be equal $G$ save for the addition of a single edge $e = (l,m)$.  \cite{dobra2011} then propose a new $\tilde{\Lambda}$ matrix by constructing a new Cholesky decomposition matrix $\tilde{\boldsymbol\Psi}$.  This is done by first letting $\tilde{\Psi}_{ij} := \Psi_{ij}$ for every $i=j$ or $(i,j) \in G$, then drawing $\tilde{\Psi}_{lm} \sim \mathcal{N}(\Psi_{lm}, \sigma_g^2)$, and finally filling in the remaining elements of $\tilde{\boldsymbol\Psi}$ according to the process described by Equation \ref{eq:cholesky}, using the new proposed graph structure $G^e$.  To accept this new element, we use the following Metropolis-Hastings ratio,
\begin{equation}
    \frac{f(\mathbf{Z} | \tilde{\Lambda}) p(\tilde{\Lambda} | G^e) p(G^e)}{f(\mathbf{Z} | \Lambda) p( \Lambda | G ) p(G)} * \frac{J(\tilde{\Lambda} \to \tilde{\boldsymbol\Psi} )}{J(\Lambda \to \boldsymbol\Psi)} * \frac{\sigma_g \sqrt{2\pi} J(\boldsymbol\Psi \to \tilde{\boldsymbol\Psi})}{\exp(-(\tilde{\Psi}_{lm} - \Psi_{lm})^2 / (2\sigma^2_g))},
\end{equation}
where $J(\cdot)$ is the Jacobian operator.  Notice that in the above calculation, the transition probability to the new state comes in the form of a normal PDF of $\mathcal{N}(\Psi_{lm}, \sigma_g^2)$, while the return probability is 1, since the process described in Equation \ref{eq:cholesky} is deterministic for a removal of an edge in the graph space.  The term $J(\boldsymbol\Psi \to \tilde{\boldsymbol\Psi})$ is the transdimensional Jacobian term described by \cite{green1995}, while the remaining Jacobian terms are for the change of variables that occurs between the precision matrix and its Cholesky decomposition.  \cite{roverato2002} provides the explicit form for each of these Jacobian terms, and the ratio thereby reduces to
\[\alpha_{DL} = \sigma_g \sqrt{2\pi} \Psi_{ll} \frac{I_{G}(\nu, D)}{I_{G^e}(\nu, D)} * \exp \left \{ -\frac{1}{2} 
\left\langle \tilde{\Lambda} - \Lambda, D + \sum_{j=1}^n Z_j Z_j' \right\rangle + \frac{ (\tilde{\Psi}_{lm} - \Psi_{lm})^2}{2 \sigma_g^2} \right\}, \] 
where $\langle A,B \rangle = \tr(A'B)$.  The determinant terms from the Wishart distributions drop out in the above ratio because $|\tilde{\Lambda}| = |\Lambda|$, which is apparent since the matrices $\tilde{\boldsymbol\Psi}$ and $\boldsymbol\Psi$ have the same diagonal elements.

The major drawback of the \cite{dobra2011} algorithm is the lingering ratio of $G$-Wishart normalizing constants: $\frac{I_{G}(\nu, D)}{I_{G^e}(\nu, D)}$.  Although there does exist an exact solution to this ratio \citep{uhler2018}, for practical purposes it must be estimated using a Laplace approximation \citep{lenkoski2011} or by an additional MCMC algorithm \citep{atay-kayis2005}.  To circumvent this issue, \cite{lenkoski2013} proposes a Double Reversible Jump algorithm, which is a variation on the Double MH exchange algorithm of \cite{wang2012} combined with the above reversible jump algorithm by \cite{dobra2011}.  

The exchange algorithm, which was first introduced by \cite{murray2006}, proposes a clever proposal distribution for a new graph structure $G$ that causes the normalizing constants $I_G(\cdot)$ drop out.  Suppose that $G^d$ is the current graph structure at the $d$th step of the MCMC process. The central idea of the exchange algorithm is that the \textit{proposal} of $G^d$ to $G^{d+1}$ involves \textit{exchanging} $G^d$ \textit{for} $G^{d+1}$ using an auxiliary variable.  In this direction, rather than consider a move starting at $(\Lambda^d, G^d)$, consider starting at the following vector with two added auxiliary variables:
\begin{equation}\label{eq:extendedVars}
    (\Lambda^d, G^d, \hat{\Lambda}^d, \hat{G}^d ).
\end{equation}
We will update this parameter vector by first performing a conjugate update on $\Lambda^d$.  Since this proposal can be done efficiently and is accepted always, this parameter is updated first, separate of the other parameters. New values for the remaining parameters are proposed using the following two moves, as first described by \cite{wang2012}:
\begin{enumerate}
    \item \label{itm:updateA} Update $\hat{\Lambda}^d, \hat{G}^d$;
    \item \label{itm:updateG} Update $G^d$ (this involves updates to the precision matrices);
\end{enumerate}

Assuming that $\Lambda^d$ has already been updated, moves \ref{itm:updateA} and \ref{itm:updateG} are done using the reversible jump framework of \cite{dobra2011} and the variation on \cite{murray2006}'s exchange algorithm given by \cite{wang2012}.

We perform move \ref{itm:updateA} by updating $\hat{\Lambda}^d$ and $\hat{G}^d$ given the current (non-auxiliary) graph structure $G^d$. We draw a new $\hat{\Lambda}^{d + 1/2}$ from the prior distribution $\mathcal{W}(\nu, D)$ and a new $\hat{G}^{d + 1/2}$ from a distribution that has the restriction that there can be at most one edge different from $G^d$.  

The the crucial idea of this approach is in its proposal scheme for move \ref{itm:updateG}.  Rather than draw a new $G^{d+1}$ randomly, we simply exchange $G^d$ and $\hat{G}^{d+1/2}$ in the full parameter vector.  Thus, $G^{d+1} := \hat{G}^{d+1/2}$ and $\hat{G}^{d+1}:=G$.  Note that this is a symmetric proposal, and that this ``exchange'' step will therefore not show up in the final MH ratio.  Since this step causes $\Lambda^{d}$ to no longer align with $G^{d+1}$, and $\hat{\Lambda}^{d + 1/2}$ to no longer align with $\hat{G}^{d+1}$, we perform the reversible jump update described above twice, perturbing each matrix to align with the correct graph structure, which gives the final vector 
\[ (\Lambda^{d+1}, G^{d+1}, \hat{\Lambda}^{d+1}, \hat{G}^{d+1} ). \]
This final alignment uses the reversible jump algorithm of \cite{dobra2011} and replaces the need to calculate the conditional distributions of the precision matrices with their changed elements integrated out, as described by the original Double MH algorithm of \cite{wang2012}.  \cite{wang2012} note that since the augmented joint distribution $p(\Lambda, G, \hat{\Lambda}, \hat{G} | \mathbf{Z})$ is equal to $p(\Lambda, G | \mathbf{Z}) q(\hat{G} | \Lambda, G, \mathbf{Z}) p(\hat{\Lambda} | \hat{G})$ where $p(\Lambda, G | \mathbf{Z})$ is the target distribution, $q(\hat{G} | \Lambda, G, \mathbf{Z})$ is a proposal distribution such that $\hat{G}$ differs from $G$ by only one edge, and $ p(\hat{\Lambda} | \hat{G})$ is $\hat{G}$-Wishart, using the extended joint distribution leaves the original posterior for $(\Lambda, G)$ marginally unaffected.

Following this process, let us return to the initial example of moving to the new graph $G^{d+1} := G^e$, which is assumed to be equal to a current graph $G$ save for the addition of a single edge $e = (l,m)$.  Let $\boldsymbol\Psi^{d}$ be the upper triangular matrix from the Cholesky decomposition of matrix $\Lambda^d$.  Further, let $\hat{\Lambda}_{G}^d$ be the matrix that, when completed using the process of \cite{roverato2002}, using the matrix $G$, gives the matrix $\hat{\Lambda}^d$.  Further, let $\hat{\boldsymbol\Psi}_G^{d}$ be its corresponding upper triangular matrix from its Cholesky decomposition.  The full MH ratio would then be 
\begin{align*}
    \frac{f(\mathbf{Z} | \Lambda^{d+1})  p(\Lambda^{d+1} | G^e)p(G^e) }{f(\mathbf{Z} | \Lambda^d) p( \Lambda^d | G ) p(G)  }  *& \frac{ J(\Lambda^{d+1} \to \boldsymbol\Psi^{d+1} )}{J( \Lambda^{d+1/2} \to \boldsymbol\Psi^{d+1/2} )} \frac{ J( \hat{\Lambda}_{G^e}^{d} \to \hat{\boldsymbol\Psi}_{G^e}^{d} ) }{ J(\hat{\Lambda}^{d+1/2} \to \hat{\boldsymbol\Psi}^{d+1/2} ) }   * \\
    &~~~\frac{\sigma_g \sqrt{2\pi} J(\boldsymbol\Psi^d \to \tilde{\boldsymbol\Psi}^{d+1} ) J(  \hat{\boldsymbol\Psi}^d \to \hat{\boldsymbol\Psi}^{d+1} )  }{\exp\left(\frac{-(\Psi^{d+1}_{lm} - \Psi^d_{lm})^2}{ 2\sigma^2_g} + \frac{( \hat{\Psi}^{d}_{lm} - (\hat{\Psi}_{G^e}^d)_{lm} )^2}{ 2\sigma^2_g}  \right)} * \frac{f_{\mathcal{W}_{G^e}(\nu, D)}( \hat{\Lambda}_{G}^{d} ) }{ f_{\mathcal{W}_{G}(\nu, D)}(\hat{\Lambda}^{d+1/2}) }
\end{align*} 

The name, Double Reversible Jump, is attributed to the dual MH steps that occur (moves \ref{itm:updateA} and \ref{itm:updateG}) and the use of \cite{dobra2011}'s reversible jump algorithm to align the swapped graph structures with their precision matrices.  The final MH ratio can be simplfied to get
\begin{equation} \label{eq:alphaDRJ}
\alpha_{DRJ} = \frac{\exp \left \{ -\frac{1}{2} \left\langle \tilde{\Lambda} - \Lambda, D + \sum_{j=1}^n Z_j Z_j' \right\rangle  + \frac{-(\Psi^{d+1}_{lm} - \Psi^d_{lm})^2}{ 2\sigma^2_g} \right\}}{ \exp \left \{ -\frac{1}{2} \left\langle \hat{\Lambda}^{d+1/2} - \hat{\Lambda}_{G}^{d} , D \right\rangle +  \frac{( \hat{\Psi}^{d}_{lm} - (\hat{\Psi}_{G^e}^d)_{lm} )^2}{ 2\sigma^2_g} \right\} } * \frac{\Psi_{ll}^d}{(\hat{\Psi}_{G^e}^d)_{ll}}.
\end{equation}

The addition to the Double Reversible Jump algorithm comes in the form of an update to the mean vector $\mu$.  Due to the choice of the N-GW distribution prior on $(\mu, \Lambda)$, we have conjugate updates for $(\mu, \Lambda)$ whenever the graph structure $G$ is fixed.  Thus, we can redraw $\mu$ along with $\Lambda$ during the initial conjugate update of $\Lambda$, which is an update that we will accept with probability 1.  Since we are interested in the extended parameter space that includes the mean vector, we must update the kernel of the above Metropolis-Hastings ratio slightly to get
\[
\frac{f(\mathbf{Z} | \tilde{\mu}, \tilde{\Lambda}) p(\tilde{\mu} | \tilde{\Lambda})  p(\tilde{\Lambda} | G^e)p(G^e) }{f(\mathbf{Z} |\mu,  \Lambda) p(\mu| \Lambda) p( \Lambda | G ) p(G)}.
\]
To handle this, we also sample a second new $\mu$ vector from the posterior, given the precision matrix $\lambda^{d+1}$, which is calculated after move \ref{itm:updateG}.  Since the proposal distribution for $\mu$ is the conjugate posterior, the MH ratio reduces back to $\alpha_{DRJ}$.

\section{Merge-Split-Swap Algorithms}\label{a:alg}
\begin{algorithm} \label{alg:splitmerge}
 \KwIn{$\boldsymbol\Theta$, parameter vectors for each regime \break $\mathbf{Z}$, the latent data \break $n$, the number of timepoints \break
 $P$, current estimate of transition probabilities \break
 $\boldsymbol\phi$, current state vector}
 \KwOut{$\tilde{\boldsymbol\phi}$, updated}
 $\mathsf{M} \leftarrow $ number of regimes in $\phi$\;
$\mathsf{perform\_split} \leftarrow \text{Binomial}(0.5)$\;
\eIf(\commentSt*[f]{Split}){$\mathsf{perform\_split}$}{
        $t \leftarrow$ selected uniformly from $\{1,\dots, M\}$\;
		Calculate distribution of split points on regime $t$, $\mathcal{P}_{\star}$, using \textbf{Algorithm \ref{alg:phiproposal}}\;
		$r_t \leftarrow$ number of timepoints in regime $t$\;
		$\mathsf{lower}_{t+1} \leftarrow$ split point $i \in \{1, \dots, r_t - 1 \}$ simulated using $\mathcal{P}_{\star}$\;
		$\alpha \leftarrow 1 / \mathcal{P}_{\star}(\mathsf{lower}_{t+1})$\;
		\For{$j \leftarrow 1, \dots, M$}{ 
		    $\hat{\phi}_j \leftarrow \begin{cases}
    \phi_{j} & \text{if } j \leq i \\
    \phi_{j} + 1 & \text{if } j > i
    \end{cases}$~\;}
  	 }(\commentSt*[f]{Merge}){
      	 $t \leftarrow$ selected uniformly from $\{1,\dots, M - 1\}$\;
  	 Calculate distribution of split points for the union of regimes $t$ and $t+1$, $\mathcal{P}_{\star}$, using \textbf{Algorithm \ref{alg:phiproposal}}\;
		$\alpha \leftarrow \mathcal{P}_{\star}(t)$\;
		\For{$j \leftarrow 1, \dots, M$}{
		$\hat{\phi}_j \leftarrow \begin{cases}
    \phi_{j} & \text{if } j \leq i \\
    \phi_{j} - 1 & \text{if } j > i
    \end{cases}$~\;
    }
  	 }
Perform Metropolis-Hastings updating step on $\boldsymbol\Theta, \phi, P, \boldsymbol\gamma$ using $\alpha$ in the ratio as the sampling probability of $\hat{\boldsymbol\phi}$ and let $\tilde{\boldsymbol\phi}$ be the accepted value\;
 \For(\commentSt*[f]{Swap}){$i \leftarrow 2 \dots n-1$}{
 $r_t \leftarrow$ number of timepoints in regime $t$\;
 \If{$\tilde{\phi}_{i-1}$ not equal $\tilde{\phi}_{i+1} \textbf{ \upshape and } r_t > 1 $}{
    $\hat{\boldsymbol\gamma}^i \leftarrow$ conjugate update of component assignments assuming $\tilde{\phi}_{i}:=\tilde{\phi}_{i-1}.$\;
    $\dot{\boldsymbol\gamma}^i \leftarrow$ conjugate update of component assignments assuming $\tilde{\phi}_{i}:=\tilde{\phi}_{i+1}.$
    $\tilde{\phi}_{i} \leftarrow \begin{cases}
    \tilde{\phi}_{i-1} & \text{w.p. }\propto P_{i-1,i-1} * f(\mathbf{z}_i | \boldsymbol\theta^{\tilde{\phi}_{i-1}}, \hat{\boldsymbol\gamma}^i) \\
    \tilde{\phi}_{i+1} & \text{w.p. }\propto (1 - P_{i-1,i-1}) * f(\mathbf{z}_i | \boldsymbol\theta^{\tilde{\phi}_{i+1}}, \dot{\boldsymbol\gamma}^i) \\
    \end{cases}$~\;
 }
 }
 \caption{Split-Merge-Swap sampling for $\boldsymbol\phi$.}
\end{algorithm}

\begin{algorithm} \label{alg:phiproposal}
 \KwIn{$\mathsf{\boldsymbol\theta}^t$, set of parameters for regime $t$  \break
 $\mathsf{lower}_t, \mathsf{upper}_t \leftarrow$ the minimum and maximum timepoints of regime $t$ \break
 $\{\mathbf{z}^{\mathsf{lower}_t}, \mathbf{z}^{\mathsf{lower}_t+1}, \dots, \mathbf{z}^{\mathsf{upper}_t}\}$, data in regime $t$}
 \KwOut{$P_\star$, probability mass function of possible split points}
$f_\mathbf{Z}(\theta)\leftarrow$ multivariate normal likelihood function of $\theta$, given data $\mathbf{Y}$\;
 $r_t \leftarrow$ number of timepoints in regime $t$\;
 $k_t \leftarrow \lfloor r_t / 10 \rfloor$\;
 Simulate $\boldsymbol\theta^{\mathsf{lower}}$ using $\{\mathbf{z}^{\mathsf{lower}_t}\}$\;
 Simulate $\boldsymbol\theta^{\mathsf{upper}}$ using $\{\mathbf{z}^{\mathsf{lower}_t + 1},\dots, \mathbf{z}^{\mathsf{upper}_t}\}$\;
$p_{q_1} \leftarrow \phi_{\{\mathbf{z}^{\mathsf{lower}_t}\}}(\boldsymbol\theta^{\mathsf{lower}} ) * \phi_{\{\mathbf{z}^{\mathsf{lower}_t + 1},\dots, \mathbf{z}^{\mathsf{upper}_t}\} }(\boldsymbol\theta^{\mathsf{upper}} )$\;
\For{$j \leftarrow 2, \dots, k_t$}{
    Sample $q_j$ from $\{\mathsf{lower}_t+(j-1)*10+1, \dots, \mathsf{upper}_t+\min\{r_t, j*10\} \}$\;
    Simulate $\theta_{\mathsf{lower}}$ using $\{\mathbf{z}^{\mathsf{lower}_t}, \dots, \mathbf{z}^{q_j-1}\}$\;
     Simulate $\theta_{\mathsf{upper}}$ using $\{\mathbf{z}^{q_j},\dots, \mathbf{z}^{\mathsf{upper}_t}\}$\;
    $p_{q_j} \leftarrow f_{\{\mathbf{z}^{\mathsf{lower}_t}, \dots, \mathbf{z}^{q_j -1 }\} }(\boldsymbol\theta^{\mathsf{lower}} ) * f_{\{\mathbf{z}^{q_j},\dots, \mathbf{z}^{\mathsf{upper}_t}\} }(\boldsymbol\theta^{\mathsf{upper}} )$\;
    Linearly interpolate values $p_{\mathsf{lower}_t + (j-1)*10 + 1}, \dots, p_{\mathsf{lower}_t + j*10 - 1}$ using $p_{\mathsf{lower}_t + (j-1)*10}$ and $p_{\mathsf{lower}_t + j*10}$
}
$P_\star \leftarrow \{p_{\mathsf{lower}_t}, p_{\mathsf{lower}_t + 1}, \dots, p_{\mathsf{upper}_t}\}$\;
 \caption{Estimating split point distribution $P_\star$.}
\end{algorithm}

\begin{algorithm} \label{alg:splitmerge_comps}
 \KwIn{$Z$, the latent data \break
 $\Theta$, the parameter sets \break
 $\Pi^r$, component probabilities for regime $r$ \break
 $N_G$, number of Gibbs sweeps for the launch vector \break
 $\boldsymbol\gamma$, current component assignments}
 \KwOut{$\boldsymbol\gamma$, updated}
 $\boldsymbol\Gamma^r \leftarrow$ elements of $\boldsymbol\gamma$ associated with regime $r$\;
 $l^r \leftarrow$ length of the vector $\boldsymbol\Gamma^r$\;
 $\mathsf{a},\mathsf{b} \leftarrow $ two unique elements of $\{1,\dots, l^r\}$\;
     $\mathsf{log\_prob\_forward}, \mathsf{log\_prob\_backwards} \leftarrow 0$\;
\eIf(\commentSt*[f]{Split}){$\boldsymbol\Gamma^r_\mathsf{a} == \boldsymbol\Gamma^r_\mathsf{b}$}{
       $\mathsf{launch\_vector} \leftarrow$ subset of $\boldsymbol\Gamma^r$ containing only values equal to $\boldsymbol\Gamma^r_\mathsf{a}$ and $\boldsymbol\Gamma^r_\mathsf{b}$\;
       Set all elements of $\mathsf{launch\_vector}$ that are equal to $\boldsymbol\Gamma^r_\mathsf{b}$ to $ \max\{\boldsymbol\Gamma^r\} + 1$\;
       $Z^\mathsf{launch} \leftarrow$ data values associated with elements of $\mathsf{launch\_vector}$\;
       $l^\mathsf{launch} \leftarrow$ number of elements in $\mathsf{launch\_vector}$\;
       \For{$j \leftarrow 1 \dots l^\mathsf{launch}$}{
        $\mathsf{closest\_value} \leftarrow \argmin_{x\in \{\mathsf{a}, \mathsf{b}\}} ||Z^r_j - Z^r_x||_2$\;
        $\mathsf{launch\_vector}_j \leftarrow  \boldsymbol\Gamma_\mathsf{closest\_value}$\;
       }
       $\mathsf{count} \leftarrow 1$\;
       \While{$\mathsf{count} \leq N_G$}{
       $\mathsf{count} \pluseq 1$\;
        \For{$j \leftarrow 1 \dots l^\mathsf{launch}$}{
        $\mathsf{launch\_vector}_j \leftarrow \begin{cases}
            \boldsymbol\Gamma^r_\mathsf{a} & \text{w.p. } \propto \pi^r_{\boldsymbol\Gamma^r_\mathsf{a}} * f(Z^\mathsf{launch}_{j} | \theta^r_{\boldsymbol\Gamma^r_\mathsf{a}} ) \\
            \boldsymbol\Gamma^r_\mathsf{b} & \text{w.p. } \propto  \pi^r_{\boldsymbol\Gamma^r_\mathsf{b}} * f(Z^\mathsf{launch}_{j} | \theta^r_{\boldsymbol\Gamma^r_\mathsf{b}} ) \\
            \end{cases} $\;
           }
       }
    $\mathsf{proposal\_vector} \leftarrow \mathsf{launch\_vector}$\;
       \For{$j \leftarrow 1 \dots l^\mathsf{launch}$}{
            $\mathsf{proposal\_vector}_j \leftarrow \begin{cases}
    \boldsymbol\Gamma^r_\mathsf{a} & \text{w.p. } \propto \pi^r_{\boldsymbol\Gamma^r_\mathsf{a}} * f(Z^\mathsf{launch}_{j} | \theta^r_{\boldsymbol\Gamma^r_\mathsf{a}} ) \\
    \boldsymbol\Gamma^r_\mathsf{b} & \text{w.p. } \propto  \pi^r_{\boldsymbol\Gamma^r_\mathsf{b}} * f(Z^\mathsf{launch}_{j} | \theta^r_{\boldsymbol\Gamma^r_\mathsf{b}} ) \\
    \end{cases} $\;
    $\mathsf{\log\_prob\_forward} \pluseq \log(\pi^r_{\mathsf{proposal\_vector}_j} * f(Z^\mathsf{launch}_{j} | \theta^r_{\mathsf{proposal\_vector}_j} )) $\;
           }
  	 }(\commentSt*[f]{Merge}){
       $\mathsf{launch\_vector} \leftarrow$ calculated following the process for \texttt{Split}, setting all elements equal to $\Gamma_\mathsf{b}^r$ to $\Gamma_\mathsf{a}^r$, and relabeling components\;
    $\mathsf{proposal\_vector} \leftarrow \mathsf{launch\_vector}$\;
       \For{$j \leftarrow 1 \dots l^\mathsf{launch}$}{
            $\mathsf{proposal\_vector}_j \leftarrow \begin{cases}
    \boldsymbol\Gamma^r_\mathsf{a} & \text{w.p. } \propto \pi^r_{\boldsymbol\Gamma^r_\mathsf{a}} * f(Z^\mathsf{launch}_{j} | \theta^r_{\boldsymbol\Gamma^r_\mathsf{a}} ) \\
    \boldsymbol\Gamma^r_\mathsf{b} & \text{w.p. } \propto  \pi^r_{\boldsymbol\Gamma^r_\mathsf{b}} * f(Z^\mathsf{launch}_{j} | \theta^r_{\boldsymbol\Gamma^r_\mathsf{b}} ) \\
    \end{cases} $\;
    $\mathsf{\log\_prob\_backwards} \pluseq \log(\pi^r_{\mathsf{proposal\_vector}_j} * f(Z^\mathsf{launch}_{j} | \theta^r_{\mathsf{proposal\_vector}_j} )) $\;
           }
  	 }
Perform Metropolis-Hastings updating step on $\boldsymbol\theta, \phi$, using $\mathsf{\log\_prob\_backwards}$ and $\mathsf{\log\_prob\_forwards}$\;
 \For(\commentSt*[f]{Swap}){$j \leftarrow 1 \dots l^\mathsf{launch}$}{
            $\mathsf{proposal\_vector}_j \leftarrow \begin{cases}
    1 & \text{w.p. } \propto \pi^r_{1} * f(Z^\mathsf{launch}_{j} | \theta^r_{1} ) \\
    2 & \text{w.p. } \propto \pi^r_{2} * f(Z^\mathsf{launch}_{j} | \theta^r_{2} ) \\
    &\vdots \\
    Q & \text{w.p. } \propto \pi^r_{Q} * f(Z^\mathsf{launch}_{j} | \theta^r_{Q} ) \\
    \end{cases} $\;
           }
 \caption{Split-Merge-Swap sampling for $\boldsymbol\gamma^r$.}
\end{algorithm}

\end{appendix} 

\end{document}